\documentclass{article}
\usepackage{authblk}

\usepackage{times}
\usepackage{amsmath,amsthm}    
\usepackage{amsfonts,bbm,bm}   
\usepackage{amssymb}
\usepackage{graphicx}
\usepackage{overpic}
\usepackage{caption,subcaption}
\usepackage{longtable}
\usepackage{setspace}
\usepackage{lineno}
\usepackage{color}
\usepackage[utf8]{inputenc}
\usepackage{xspace}

\usepackage{microtype}     
\usepackage{hyperref}      

\usepackage{multirow}
\usepackage[table,xcdraw]{xcolor}
\usepackage{eurosym}
\usepackage[final]{changes}

\newtheorem{theorem}{Theorem}[section]
\newtheorem{lemma}[theorem]{Lemma}

\theoremstyle{definition}
\newtheorem{definition}[theorem]{Definition}

\theoremstyle{remark}
\newtheorem{remark}[theorem]{Remark}

\numberwithin{equation}{section}

\newcommand{\bfs}{\mathbf{s}} 

\newcommand{\bfr}{\mathbf{r}}
\newcommand{\Do}{\mathcal{D}} 
\newcommand{\pDo}{\partial\mathcal{D}}
\newcommand{\om}{\omega}

\newcommand{\xk}{x^{\ast}}

\newcommand{\EE}{\mathbb{E}} 
\newcommand{\E}{\mathrm{e}} 

\newcommand{\Prec}{\mathbf{Q}}

\newcommand{\Na}{\mathbb{N}}
\newcommand{\R}{\mathbb{R}}
\newcommand{\Rd}{\mathbb{R}^{d}}

\newcommand{\Or}{\mathcal{O}}
\newcommand{\cov}{\bm{\Sigma}}
\newcommand{\pa}{\partial}
\newcommand{\bfo}{\mathbf{0}}
\newcommand{\bfx}{\mathbf{x}}
\newcommand{\bfk}{\mathbf{k}}
\newcommand{\kk}{\lVert \mathbf{k} \rVert}
\newcommand{\bfC}{\mathbf{C}}
\newcommand{\ff}{\rho}
\newcommand{\He}{\mathrm{He}}

\newcommand{\rr}{\|{\bfr}\|} 

\newcommand{\Ha}{{\mathcal H}}  

\newcommand{\Grid}{{\mathcal{G}}}
\newcommand{\ft}{{\mathrm{FT}}}
\newcommand{\ift}{{\mathrm{IFT}}}
\newcommand{\D}{{\mathrm{d}}}
\newcommand{\I}{\jmath}
\newcommand{\bmthe}{{\bm \theta}}  
\newcommand{\bmal}{{\bm \alpha}}  

\newcommand{\Samp}{{\mathbb{S}_N}}
\newcommand{\Kse}{K_{{\mathrm{se}}}}
\newcommand{\Kseft}{\tilde{K}_{{\mathrm{se}}}}

\newcommand{\la}{\lambda}

\newcommand{\beq}{\begin{equation}}
\newcommand{\eeq}{\end{equation}}

\newcommand{\bit}{\begin{itemize}}
\newcommand{\eit}{\end{itemize}}

\begin{document}

\title{Boltzmann-Gibbs Random Fields with Mesh-free Precision Operators  Based on Smoothed Particle Hydrodynamics}

\author{Dionissios~T.~Hristopulos\thanks{\texttt{dchristopoulos@ece.tuc.gr};  Corresponding author}\hspace{2pt}}

\affil{School of Electrical and Computer Engineering, \\Technical University of Crete, Chania, 73100 Greece}

\maketitle


\begin{abstract}
Boltzmann-Gibbs random fields are defined in terms of the exponential expression $\exp\left(-\Ha\right)$, where $\Ha$ is a suitably defined
energy functional of the field states $x(\bfs)$. This paper presents a new Boltzmann-Gibbs  model which features local interactions in
the energy functional. The interactions are embodied in a spatial coupling function which uses smoothed kernel-function approximations of spatial derivatives inspired from the theory of smoothed particle hydrodynamics.  A specific model for the interactions based on a second-degree polynomial of the  Laplace operator is studied. An explicit, mesh-free  expression  of  the spatial coupling function (precision function) is derived for the case of the squared exponential (Gaussian) smoothing kernel. This coupling function allows the model to seamlessly extend from  discrete data vectors to continuum fields.   Connections with Gaussian Markov random fields and the Mat\'{e}rn field with $\nu=1$ are established.
\end{abstract}

\maketitle

\newpage

\section{Introduction}
The particle-wave duality is a central concept of quantum mechanics. It refers to the fact that physical entities can be described as either waves (continuum functions) or particles (quanta), and each viewpoint has advantages depending on the particular phenomenon studied. For example, in studies of light propagation it is more suitable to view light as an electromagnetic wave; in contrast, in the photoelectric effect,  light is best modeled as discrete particles (photons).   There is a similar duality in probability and statistics: Random fields are defined over continuum support spaces,  extending throughout an entire spatial domain, like waves. On the other hand, data-driven applications focus on discretely sampled states, in which  the measured values of the field bear resemblance to discrete particles. The field values  at unmeasured points then need to be inferred (in terms of a  marginal probability distribution conditioned on the data) in order to obtain a complete representation of the studied process  in continuum space.

For Gaussian random fields defined in terms of the expectation and the covariance function it is straightforward to reconcile the continuum and discrete viewpoints: if the expectation and covariance functions are known, the joint probability density function is known for every countable set of points.  In addition, the marginal distribution of the field states is known  at every point in the domain of definition. Furthermore, the \emph{consistency (marginalization)} property of the covariance function implies that if $S_{1} \subset S$ are two sets of sampling points, the covariance matrix $\cov_{1}$ between any pairs of points in $S_{1}$ is given by the submatrix of $\cov$ corresponding to $S_{1}$.


In contrast to the covariance, the consistency property does not hold if the spatial correlations are defined in terms of the inverse covariance (precision) matrix~\cite{Rasmussen06}.
To illustrate this point, consider  the paradigmatic model of the \emph{free Gaussian field theory} used in statistical physics~\cite{Mussardo10}; a rigorous mathematical treatment of this model is given in~\cite{Sheffield07}. The energy functional of the free Gaussian field over a compact domain $\Do \subset \Rd$ is given by
\beq
\label{eq:free-field}
\Ha = c \int_{\Do} \D\bfs \, \sum_{i=1}^{d} \left( \frac{\pa x(\bfs)}{\pa s_{i}}\right)^2 \, .
\eeq

 In the free Gaussian  model, the spatial dependence is introduced via the field's spatial derivatives $\{ \pa x(\bfs) / \pa s_{i} \}_{i=1}^{d}$. In order to determine the energy for a discrete set of positions, e.g., over the nodes of a square grid $\{ \bfs_{n} \}_{n=1}^{N} = \Grid$,  the derivatives are approximated by means of finite differences. This approximation alters the short-range behavior of the model compared to the continuum~\cite{dth03}.  Assuming a forward-difference  discretization of the partial derivatives on $\Grid$, the energy function of Eq.~\eqref{eq:free-field} can be expressed as
 $\Ha = \frac{1}{2} \sum_{i,j }x_{i} Q_{i,j} x_{j}$, where the double summation is over all sites $\bfs_{i}, \bfs_{j} \in \Grid$,  and the entries $Q_{i,j}$ of the precision matrix are non-zero  if either  $\bfs_{i}=\bfs_{j}$ or if $\bfs_{i}$ and $\bfs_{j}$ are nearest neighbors~\cite{ejs21}.

If we now consider a subset $\Grid' \subset \Grid$ of grid nodes, the respective precision matrix $\Prec'$ is not simply a submatrix of $\Prec$.  As an example, consider a set of  sites along the one-dimensional real line $\Grid: s_{n}=n a$, $n= 1, \ldots, 2N$, where $N\in \Na$.  Using forward difference discretization, the precision matrix $\Prec$ couples each point $s_{n}$ with its nearest neighbors, $s_{n \pm 1}$\footnote{Assuming periodic boundary conditions, $s_{2N+1}=s_{1}$ and $s_{0}=s_{2N}$.}.  Next, consider the reduced system  obtained by dropping every second node of $\Grid$. To obtain the coupling between the remaining nodes (which are uncoupled in the initial system), one needs to marginalize the joint distribution over the nodes that are removed.  The resulting precision matrix $\Prec'$ of the reduced system is not a submatrix of $\Prec$.

For a set of scattered data points $\Samp$~\cite{Franke82b}, the discretization of the derivatives is an additional issue. Since the sampling positions are irregularly spaced, there is no straightforward approximation for the spatial derivatives.  It is possible to construct spatial models which mimic the spatial derivatives by means of weighted spatial increments~\cite{dth15c,dth19}.  However, such models do not satisfy the consistency property and thus cannot bridge the continuum and discrete viewpoints in a straightforward manner.


In this paper we attempt to address the discretization and consistency issues following an approach  inspired by \emph{smoothed particle hydrodynamics (SPH)}~\cite{Monaghan92,Monaghan05}. SPH is a computational mesh-free Lagrangian method used for simulating physical phenomena in continuum media. Since SPH does not rely  on a numerical grid,  it is ideal  for processes which involve complex boundary dynamics. Our interest in SPH is stimulated by its ability to model spatial derivatives in terms of smoothing kernel functions without the requirement for an underlying numerical grid.  We take advantage of this idea to formulate a Boltzmann-Gibbs Gaussian random field with an explicitly defined \emph{spatial coupling function (precision function)}.  The latter determines the interactions between the field values at different locations that are arbitrarily distributed in a domain $\Do$. The random field thus defined extends the notion of Gauss-Markov random fields by introducing a mesh-free, distance-dependent precision function which can be used to derive a respective precision matrix for any configuration of sampling points.

The main results presented in this paper are as follows: Theorem~\ref{theorem:gaussian} proves that a precision operator given by the \emph{second-degree polynomial of the Laplace operator} (LAP2) defines a Gaussian measure.  Theorem~\ref{theorem:sph-lap2-disc}  establishes a new Gaussian random field using LAP2 and the SPH representation.  Theorem~\ref{theorem:sph-lap2-prec}  derives the LAP2-SPH \emph{spatial coupling (precision) function} in terms of the kernel-based, spatial interaction smoothing function and its derivatives. Finally,   Theorem~\ref{theorem:sph-lap2-prec-gauss}  obtains an explicit expression for the LAP2-SPH precision function if a Gaussian  smoothing kernel  is applied.

The structure of the paper is as follows: Section~\ref{sec:prelim} presents  definitions used in subsequent sections. Section~\ref{sec:gauss-bg} defines a Gaussian measure based on the second-degree polynomial of the Laplace operator. In Section~\ref{sec:gauss-bg-kernel} this construction is modified to include  kernel smoothing, leading to the LAP2-SPH representation. The extension of the SPH formulation to the continuum is demonstrated in Section~\ref{sec:sph-lap2-continuum}. In Section~\ref{sec:discussion} connections of the theory and possibilities for future research are discussed.  Finally, Section~\ref{sec:conclusions} summarizes our conclusions.

\section{Preliminaries}
\label{sec:prelim}

\begin{definition}[Vectors and matrices]
We  use boldfaced letters, e.g., $\mathbf{A}$, to denote vectors and matrices. The symbol $\mathbf{A}^\top$ will indicate  the transpose of $\mathbf{A}$.
The position vector in $\Rd$ will be denoted by $\bfs \in \Rd$.
\end{definition}

\begin{definition}[The Laplacian and Biharmonic operators]
 $\nabla$ is the vector gradient operator with components $\nabla_{i}= \pa / \pa s_{i}$, $i=1, \ldots, d$.  The squared gradient  $\nabla^2 = \nabla \cdot \nabla =  \sum_{i=1}^{d} \pa^{2} / \pa s_{i}^{2}$,  defines the  \emph{Laplace operator}. The symbol $\triangle$ is also used to denote the Laplacian.

 The square of the Laplace operator, i.e., $\nabla^{4}= \triangle^2$, will denote the \emph{biharmonic} or \emph{Bi-Laplacian operator}.
 \end{definition}

\begin{definition}[Random Field]
Let $\bfs \in \Do \subset \Rd$ denote the position vector over a domain $\Do$ which is a bounded subset of  $\Rd$. The boundary $\pDo$ of $\Do$ is assumed to be a piecewise smooth (i.e., infinitely  differentiable) manifold of dimension $d-1$.  Given a probability space $(\Omega, \textsl{F,P})$, where ${\Omega}$ is the sample space, $\textsl{F}$ is a
$\sigma-$field of subspaces of ${\Omega}$, and $\textsl{P}$ is a probability measure, the collection of real-valued, scalar random variables  $ \{ X(\bfs,\omega):  \bfs \in \Do, \, \omega\in\Omega \}$ indexed by $\bfs$ is a scalar, real-valued \emph{spatial random field} $X: \Do \times \Omega \to \R$~\cite{Yaglom87}.
\end{definition}

The realizations of the random field $X(\bfs;\om)$ for $\omega\in\Omega$ are sample functions denoted by  $x(\bfs)$.  The \emph{expectation} over the probability measure will be denoted by $\EE[\cdot]$, e.g.,
$\EE[X(\bfs;\om)] =\int_\Omega x(\bfs; \om) \, P(\mathrm{d}\om) <\infty$.  It will be assumed that both the mean function $m(\bfs) = \EE[X(\bfs;\om)]$ and the covariance function
$C(\bfs_{1}, \bfs_{2})= \EE[X'(\bfs_{1};\om)\, X'(\bfs_{2};\om)]$, where $X'(\bfs;\om)=X(\bfs;\om)- m(\bfs)$ is the fluctuation field,  are well defined.
Without loss of generality, we focus on  $m(\bfs)=0$.

\begin{definition}[Second-order homogeneity]
A random field $X(\bfs;\om)$ is \emph{second-order homogeneous} or \emph{stationary} if (i) $m(\bfs)= c \in \R$ for all $\bfs \in \Do$ and (ii) $C(\bfs_{1},\bfs_{2})= C(\bfs_{1}-\bfs_{2})$ for all $\bfs_{1}, \bfs_{2} \in \Do$.
\end{definition}

\begin{definition}[Fourier transforms]
Let $C(\bfr): \Rd \to \R$ represent a function which is absolutely integrable over $\Rd$. Then, the   \emph{Fourier transform} (FT) and its inverse (IFT) exist. The FT is given by
\begin{equation}
\label{eq:covft}
\widetilde{C}(\bfk) := \ft[ C(\bfr)] = \int_{\Rd} \D\bfr\; \E^{-\I \bfk\cdot \bfr}  C(\bfr),
\end{equation}
where $\I=\sqrt{-1}$ is the imaginary unit and $\bfk \in \Rd$ is the spatial frequency vector (wavevector) in reciprocal (Fourier) space.

The inverse FT is given by means of the integral
\begin{equation}
\label{eq:invcovft}
C(\bfr) := \ift[ \widetilde{C}({\bfk})] = \frac{1}{(2\,\pi)^{d}}\,\int_{\Rd} \D\bfk\;
\E^{\I \,\bfk\cdot \bfr} \, \widetilde{C}(\bfk).
\end{equation}

\end{definition}

\begin{definition}[Positive-definite matrix]
A symmetric $N \times N$ matrix $\bfC$ is  positive-definite (non-negative definite)~\footnote{We use the terms non-negative definite and positive definite as equivalent; if $x^\top \bfC \bfx > 0$ for all $\bfx \in \R^{N}$ then $\bfC$ is strictly positive definite.} if   $x^\top \bfC \bfx \ge 0$ for all $\bfx \in \R^{N}$.
\end{definition}

\begin{definition}[Covariance function]
A function $C(\bfs, \bfs'): \Do \times \Do \to \R$ is a permissible covariance function for some random field $X(\bfs;\om): \Do \times \Omega \to \R$, if and only if  it is a non-negative definite function. A function  $C(\cdot, \cdot): \Do \times \Do \to \R$ is non-negative definite if for all $N \in \Na$ and for all sets of points $\{ \bfs_{i} \}_{i=1}^{N}$, where $\bfs_{i} \in \Do$, the matrix $\bfC$ with entries $C_{i,j}=C(\bfs_{i}, \bfs_{j})$ is non-negative definite.
\end{definition}

\begin{theorem}[Bochner's theorem]
A function $C(\bfr)$ is a permissible covariance for a stationary random field $X(\bfs;\om)$ if and only if the Fourier transform $\widetilde{C}({\bfk})$ of $C(\bfr)$ exists, is non-negative,  and its integral over $\Rd$ is finite~\cite{Bochner59}.
\end{theorem}

\section{Gaussian Measure for Boltzmann-Gibbs Random Fields}
\label{sec:gauss-bg}

Boltzmann-Gibbs random fields are defined by means of  exponential ``joint-density'' expressions, i.e., $\ff[x(\bfs)] = Z^{-1}\exp\left(-\Ha[x({\Do})] \right)$, where
$Z$ is a normalization constant, known as the partition function and $x({\Do})= \{ x(\bfs), \bfs \in \Do \}$ denotes a field configuration over the  domain $\Do$.
The calculation of $Z$ is not straightforward, or even mathematically  well-defined for infinite-dimensional configuration spaces (i.e., if the field states are continuous functions).
Nonetheless, this definition of Boltzmann-Gibbs random fields is  used in statistical field theory~\cite{Kardar07,Mussardo10}.  It is possible  to rigorously define   Gaussian measures in such infinite-dimensional spaces as shown in~\cite{Sheffield07}.

\begin{remark}[About the partition function]
\label{rem:Z}
The value of $Z$ is not necessary in order to calculate the relative probabilities of two different states, $x_{1}(\Do)$ and $x_{2}(\Do)$, i.e.,  $\ff[x_{1}(\Do)]/\ff[x_{2}(\Do)] = \E^{-\Ha[x_{1}(\Do)]+\Ha[x_{2}(\Do)]}$. Similarly, the  \emph{characteristic functional (CF)} can be defined without reference to $Z$.  Therefore, it is possible to define a Gaussian measure in terms of the CF without reference to the partition function~\cite{Glimm12}.

\end{remark}

This paper focuses on  B-G random fields with an  energy functional which is quadratic with respect to $x(\bfs)$.  The spatial dependence is introduced by means of a coupling function that involves  a polynomial of the   Laplace operator (see also Remark~\ref{rem:bg-fgc}).
\begin{definition}[Inner product]
\label{defi:inner}
The inner product $\langle f(\bfs), g(\bfs)\rangle$ is defined as follows:
\begin{enumerate}
\item Let $f(\bfs), g(\bfs): \Do \subset \Rd \to \R$ be square-integrable scalar functions, i.e., $f(\cdot), g(\cdot) \in L^{2}(\Do)$. Then
\begin{subequations}
\beq
\label{eq:inner-scalar}
\langle f(\bfs), g(\bfs)\rangle = \int_{\Do}\D\bfs \, f(\bfs)\, g(\bfs).
\eeq
\item Let $\mathbf{f}_{v}(\bfs), \mathbf{g}_{v}(\bfs): \Do  \to \Rd$ be square integrable vector functions with $d$ components, i.e., $\mathbf{f}_{v}(\bfs)= \left( f_{v,1}(\bfs), \ldots,  f_{v,d}(\bfs)\right)^\top$. Then,
\beq
\label{eq:inner-vector}
\langle \mathbf{f}_{v}(\bfs), \mathbf{g}_{v}(\bfs)\rangle = \int_{\Do}\D\bfs \, \mathbf{f}_{v}(\bfs)\cdot \mathbf{g}_{v}(\bfs), \eeq
\end{subequations}
 where the dot product of two vector functions is defined by $\mathbf{f}_{v}(\bfs) \cdot \mathbf{g}_{v}(\bfs)=\sum_{i=1}^{d} f_{v,i}(\bfs)\, g_{v,i}(\bfs)$.
\end{enumerate}
\end{definition}

\medskip

\begin{theorem}[Gaussian measure based on Polynomial of  Laplace operator]
\label{theorem:gaussian}
Let   ${\mathcal Q}(\bmthe)$ represent the following second-degree polynomial of the Laplace operator $\triangle = \nabla^{2}$,
\beq
\label{eq:Q}
{\mathcal Q}(\bmthe) = \theta_{0} - \theta_{1} \triangle + \theta_{2} \triangle^{2},
\eeq
where $\bmthe=(\theta_{0}, \theta_{1}, \theta_{2})^\top$  is a vector of polynomial coefficients  and $\triangle^{2} = (\nabla^2)^2$ is the Bi-Laplacian.
Let us assume that the polynomial coefficients satisfy one of the following conditions:
\begin{enumerate}
    \item[C1:] $\theta_{0}, \theta_{1}, \theta_{2} >0$.

    \item[C2:]  $\theta_{0}, \theta_{2} >0$, $\theta_{1}<0$  and $\theta_{1}^2 < 4 \theta_{0} \theta_{2}$.
\end{enumerate}
Then, the following statements are true:



\begin{enumerate}
    \item[S1:] ${\mathcal C} = {\mathcal Q}^{-1}$ defines a  covariance operator in the Sobolev space $W^{2,2}(\Do)$.

    \item[S2:] The covariance ${\mathcal C}$  defines a unique Gaussian measure with characteristic functional
\beq
\label{eq:gauss-mgf}
S(f) = \E^{-\langle f, \, {\mathcal C} \, f \rangle/2}, \; \mbox{where} \; f \in W^{2,2}(\Do), \; f\mid_{\pDo}=0,
\eeq
and $\langle \cdot\, , \cdot \rangle \, : W^{2,2}(\Do) \times W^{2,2}(\Do) \to \R$ is the $L^{2}(\Do)$ inner product.
\end{enumerate}

\end{theorem}

\begin{proof}

(S1) We consider functions $f(\cdot) \in W^{2,2}(\Do)$ so that all the second-order partial  derivatives are defined in the weak sense.  We need to show that ${\mathcal C}={\mathcal Q}^{-1}$ defines a \emph{positive-definite, continuous, non-degenerate bilinear form}  $\langle f, {\mathcal C}\,g\rangle$ on $W^{2,2}(\Do) \times W^{2,2}(\Do)$.

\medskip
(i) To prove \emph{positive-definiteness} of ${\mathcal C}$ we need to show that $\langle f, {\mathcal C}\,f\rangle \ge 0$ for all $f \in W^{2,2}(\Do)$.
The existence theorem for the eigenfunctions of the Laplace operator states that if $\Do$ is any open bounded subset of $\Rd$ with boundary $\pDo$, then there exists in the Sobolev space $W^{2,2}(\Do)$ an orthonormal  basis $\{ f_{n}\}_{n=1}^{\infty}$ of $L^{2}(\Do)$
 such that $-\nabla^{2} f_{n}=\la_{n} f_{n}$ in $\Do$ with Dirichlet boundary conditions, i.e.,  $f_{n}=0$ on $\pDo$. The eigenvalues $\la_{n}$ form a non-decreasing sequence $0 < \la_{1} \le \la_{2} \le \la_{3} \ldots$ with $\la_{n} >0$ for all $n \ge 1$, and $\la_{n} \to \infty$ as $n \to \infty$; in addition,  the $f_{n}$ are in the  space of smooth functions, $f_n \in C^{\infty}(\Do)$ and $\langle f_{n}\, , f_{m} \rangle = \delta_{n,m}$~\cite[p.~355]{Evans02}, \cite{Nica11}.

 The eigenfunctions of the  (negative) Laplace operator are also eigenfunctions of the biharmonic operator, i.e., $\nabla^{4} f_{n}=(-\nabla^{2})\, (-\nabla^{2}) f_{n} =  -\nabla^{2} \left(\la_{n} f_{n}\right) = \la_{n}^2 f_{n}$. Therefore, the eigenfunction basis of $\nabla^2$ also plays the same role for the operators ${\mathcal Q}$ and ${\mathcal C}$, i.e.,
\[
{\mathcal Q} f_{n} = \left( \theta_{0} + \theta_{1} \la_{n} + \theta_{2} \la_{n}^{2} \right)\, f_{n} \Rightarrow
{\mathcal C} \, f_{n} = \frac{f_{n}}{p_{2}(\la_{n})} \, , \; \forall n=1, 2, \ldots,
\]
where $p_{2}(\la_{n}) = \theta_{0} + \theta_{1} \la_{n} + \theta_{2} \la_{n}^{2} $ is the \emph{characteristic polynomial} of ${\mathcal Q}$.

Let $f \in W^{2,2}(\Do)$. Then,
$f$ can be expanded in the Laplacian orthonormal basis, i.e.,  $f = \sum_{n} a_{n} f_{n}$, where $a_{n} = \langle f_{n}, f \rangle$.  In light of the orthonormality of the eigenfunction basis with respect to the $L^{2}(\Do)$ inner product, it follows that
\beq
\label{eq:C-pd}
\langle f, {\mathcal C}f \rangle = \sum_{n}\sum_{m} \frac{a_{n} a_{m}}{p_{2}(\la_{n}) } \langle f_{m}, f_{n} \rangle = \sum_{n}\frac{a_{n}^{2}}{p_{2}(\la_{n}) } \, .
\eeq

The operator ${\mathcal C}$ is associated with a real-valued, symmetric bilinear functional which is defined in terms of the following $L^{2}(\Do)$  inner product
\[
\langle g, {\mathcal C} f\rangle =
\sum_{n} \frac{a_{n} b_{n}}{p_{2}(\la_{n})}, \; \mbox{where} \;f, g \in W^{2,2}(\Do), \; \mbox{and} \;
g=\sum_{m}b_{m} f_{m}.
\]
This equality  is derived from the orthonormality of the eigenfunctions.  The above definition of $\langle g, {\mathcal C} f\rangle$ guarantees that $\langle (g_{1} + g_{2}), {\mathcal C} f\rangle = \langle g_{1}, {\mathcal C} f\rangle + \langle g_{2}, {\mathcal C} f\rangle$ as well as $\langle g, {\mathcal C} (f_{1} + f_{2})\rangle = \langle g, {\mathcal C} f_{1}\rangle + \langle g, {\mathcal C} f_{2}\rangle$, thus proving the \emph{bilinearity} of the functional $\langle g, {\mathcal C} f\rangle$.

In light of the spectral expansion in Eq.~\eqref{eq:C-pd},  the \emph{positive-definiteness} of ${\mathcal C}$  is guaranteed if the characteristic polynomial $p_{2}(z) = \theta_{0} + \theta_{1} z + \theta_{2} z^2$ of the precision operator ${\mathcal Q}$ takes positive values for all $z>0$ ($z>0$ reflects that  $\lambda_{n}>0$ for the  negative Laplacian's eigenvalues).  This is  true if condition (C1) holds, i.e., $\theta_{0}, \theta_{1}, \theta_{2} >0$.  If $\theta_{1} <0$, the roots of $p_{2}(z)$ should not be real numbers; this  is guaranteed if the discriminant of $p_{2}(z)$ is negative, i.e., $\theta_{1}^{2} - 4\theta_{0}\theta_{2} <0$ as stated in (C2). Negative values of $\theta_{0}$ are not permissible, since $\theta_{0} <0$ leads to negative values of $p_{2}(z)$ for $z= 1/\delta$ where $\delta \to \infty$. Negative values of  $\theta_{2}$ are also not permissible, since they  lead to $p_{2}(z) <0$ for $z \to \infty$.

\medskip

(ii) The \emph{continuity} of the bilinear functional $\langle f\,, {\mathcal C}f\rangle$ requires that the functional is bounded~\cite{Hayden67}, i.e., that there is $\delta >0$ such that
$\langle f\,, {\mathcal C}f\rangle \le \frac{1}{\delta}\langle f\,, f\rangle$.  Without loss of generality assume that $\langle f\,, f\rangle=1$.  Based on the spectral expansion of the functional given by Eq.~\eqref{eq:C-pd}, the existence of a bound requires
\[
\sum_{n}\frac{a_{n}^{2}}{p_{2}(\la_{n}) } \le \frac{1}{\delta}\,.
\]
For this inequality to hold it is sufficient that
\[
\sum_{n}\frac{1}{p_{2}(\la_{n}) } \le \frac{1}{\delta a^{2}_{\max}},\, \mbox{where}\;
a_{\max}= \max\,\{ a_{n}\}_{n=1}^{\infty}.
\]
One idea is to investigate the convergence of the  spectral series  $\sum_{n}p_{2}^{-1}(\la_{n})$ by means of d'Alembert's ratio test, i.e., by
\[
\lim_{n \to \infty} \, \left\lvert \frac{p_{2}(\la_{n})}{p_{2}(\la_{n+1})} \right\rvert= \lim_{n \to \infty} \, \left(\frac{\la_{n}}{\la_{n+1}} \right)^{4}.
\]
However, since   Laplacian eigenvalues may have a multiplicity larger than one,   as $n \to \infty$ for certain $n$ it is possible that 
$\la_{n}/\la_{n+1}=1$, which is inconclusive for convergence.

Instead, we will use the Maclaurin-Cauchy integral test which states that a series $\sum_{n=1}^{\infty} u_{n}$, where $u_{n} \ge 0$, converges if $\int_{1}^{\infty} \phi(x) \, \D x$ is finite and diverges if the integral is infinite; $\phi(x): \R_{0,+} \to \R_{+}$ is a positive, continuous, monotonically decreasing function such that $\phi(n)=u_{n}$ and $\phi(0)$ is finite; herein,  $\phi(x) = 1/p_{2}(\la(x))$, where $x$ is a real-valued extension of the eigenvalue number. The function $\phi(\la(x))$ is positive and continuous if either condition C1 ($\theta_{0}, \theta_{1}, \theta_{2} >0$) or C2 ($\theta_{0}, \theta_{2}
 >0, \theta_{1} <0 \, \wedge \, \theta_{1}^{2} < 4 \theta_{0}\theta_{2}$) holds.

 If C1 holds, $p_{2}(\la(x))$ is monotonically increasing and thus $\phi(\la(x))$ is  monotonically decreasing.  Then,
\[
\int_{1}^{\infty} \phi(x) \, \D x =
\int_{1}^{x^\ast} \phi(x) \, \D x + \int_{x^\ast}^{\infty} \phi(x) \, \D x \,
\]
where $x^\ast \in \R_{+}$ is arbitrarily large but finite. The integral from 1 to ${x^\ast}$ converges since $\phi(x)$ does not have singularities in this interval. If the improper integral from ${x^\ast}$ to $\infty$ converges, then $\int_{1}^{\infty} \phi(x) \, \D x$ also converges. We express the improper integral as
\[
\int_{x^\ast}^{\infty} \phi(x) \, \D x =
\int_{\la^\ast}^{\infty} \D \la\, \rho(\la) \, \frac{1}{\theta_{0} + \theta_{1} \la^{2} + \theta_{2} \la^{4}},
\]
where $\rho(\la) = \D x(\la)/\D \la$ is the density of eigenvalues with value $\la$.  For this integral to converge, $\rho(\la)$ should increase more slowly than $\la^3$ as $\la \to \infty$.
According to \emph{Weyl's law}  the number of Laplacian eigenvalues $n(\la)$ with values less than or equal to $\lambda$ satisfies  asymptotically
\[
\lim_{\la \to \infty}\frac{n(\la)}{\la^{d/2}} = \frac{v_{d}\lvert \Do \rvert }{(2\pi)^d},
\]
where $v_{d}$ is the volume of the unit sphere in $d$ dimensions and $\lvert \Do \vert$ is the volume of the domain $\Do$~\cite{Weyl12}, \cite[p.~356]{Evans02}. Hence, it follows that $\rho(\la)= c_{d} \la^{d/2-1}$ where $c_{d}=\frac{v_{d}\, d\,\lvert \Do \rvert }{2(2\pi)^d}$. Therefore, the improper integral converges if $d/2<4$.

If C2 holds, $\phi(z)$ has  a  maximum at $z^\ast = -\theta_{1}/2\theta_{2}$, while it is monotonically decreasing for $z>z^\ast$. If we denote $\la^\ast=\sqrt{z^\ast}$, the finite series $\sum_{n=1}^{n^\ast}p_{2}^{-1}(\la_n)$, where $n^{\ast}= \min\left( n \in \Na: \la^\ast < \la(n^\ast) \right)$, is summable.  The convergence of the series
$\sum_{n=n^\ast +1}^{\infty}p_{2}^{-1}(\la_n)$ is established using the integral test as shown above.

\medskip
(iii) Based on Eq.~\eqref{eq:C-pd}, $\langle f, {\mathcal C} \, f \rangle=0$  is only possible  if $a_{n}=0$ for all $n$, i.e., only if $f=0$. Hence, the covariance  operator is \emph{non-degenerate}.

\vspace{6pt}

(S2) Since ${\mathcal C}={\mathcal Q}^{-1}$ is a covariance operator (as shown above),  the existence theorem~\cite[p.~100]{Glimm12}  asserts that there exists a unique Gaussian measure  with the characteristic functional given by~\eqref{eq:gauss-mgf}.

\end{proof}

\begin{remark}[BG-LAP2 Gaussian density]
\label{rem:bg-fgc}
Let us consider  square integrable functions  $x(\bfs) \in W^{2,2}(\Do)$.
Based on the CF of  Eq.~\eqref{eq:gauss-mgf}  for the Gaussian measure defined in Theorem~\eqref{theorem:gaussian}, a BG-LAP2 Gaussian joint \emph{probability density function (pdf)} can be defined as $\ff =Z^{-1} \exp\left( -\Ha[x({\Do})] \right)$, with the following energy functional:
\beq
\label{eq:Ha}
\Ha[x({\Do})] = \frac{1}{2}\, \langle x(\bfs), {\mathcal Q} x(\bfs) \rangle,
\eeq
subject to the proviso of the comments in Remark~\ref{rem:Z}. The acronym LAP2 refers to the second-degree polynomial of the  Laplace operator used to construct ${\mathcal Q}$.

Using integration by parts~\cite{Sheffield07,dth20}, the BG-LAP2 quadratic energy functional Eq.~\eqref{eq:Ha} can be expressed as follows
\beq
\label{eq:bg-fgc}
\Ha[x({\Do});\bmthe] = \frac{1}{2}\left[ \theta_{0} \langle x(\bfs), x(\bfs)\rangle  + \theta_{1}  \langle \nabla x(\bfs), \nabla x(\bfs)\rangle +  \theta_{2}  \langle \nabla^{2} x(\bfs), \nabla^{2} x(\bfs)\rangle\right].
\eeq
\end{remark}

\medskip

\section{Boltzmann-Gibbs Random Fields with Kernel-Mediated Interactions}
\label{sec:gauss-bg-kernel}

Kernel smoothing is used in SPH to approximate functions defined in continuum spaces in terms of discrete, irregularly placed  ``pseudo-particles'' which are seamlessly connected  by means of the kernel function. Below we use  \emph{kernel (window) functions} to formulate smoothed expressions for the configurations  in $\Ha[\cdot]$.  This implies a respective smoothing of the spatial coupling, leading to a mesh-free representation of spatial interactions. This section focuses on a finite-dimensional vector $\bfx$ of observations.  In this case it is  unnecessary to restrict $\bfs$ within a bounded domain $\Do$, since the spatial integration over $\Do$ in the energy functional $\Ha[\cdot]$ is replaced by a summation.  The SPH representation introduces a smoothed field $\xk(\bfs)$, where in this case $\bfs \in \Rd$. Hence,  inner products that involve the smoothed fields  represent integrals over $\Rd$.

\begin{definition}[Smoothing Kernels]
\label{defi:kernel}
A real-valued, non-negative function $K(\bfs, \bfs';h): \Rd \times \Rd \to \R_{+,0}$ and $h>0$ is a smoothing kernel function if it satisfies the following properties:
\begin{enumerate}
\item Translation symmetry: $K(\bfs, \bfs';h)=K(\bfs- \bfs';h)$.

\item Symmetry under interchange of positions: $K(\bfs, \bfs';h)=K(\bfs', \bfs;h)$.

\item Existence of all partial derivatives of order $q$ (herein it is sufficient that $q=4$).

\item Normalization condition: $\int_{\Rd} \D\bfs\, K(\bfs - \bfs';h)= 1$.

\item Maximum of Fourier transform  at zero frequency: $\tilde{K}(\bfk;h) \le \tilde{K}(\bfo;h)$ for all $\bfk \in \Rd$, where $\bfo$ is the zero vector in $\Rd$.

\item Isotropic dependence, i.e., $K(\bfr;h)=K(r;h)$ is a \emph{radial function}\footnote{For economy of notation we do not distinguish between the function $K(\cdot;h): \Rd \to \R$ and the radial function $K(\cdot;h): \R \to \R$.} which depends on $\bfr$ exclusively through the Euclidean distance $r=\rr$.
\end{enumerate}

\end{definition}

\begin{remark}[Zero-frequency peak]
Property (5) assumes the existence of the  Fourier transform $\tilde{K}(\bfk;h)$. The latter exists because $K(\bfs-\bfs';h)$ is a non-negative function which is integrable due  to property (4); hence, it is also absolutely integrable.
\end{remark}

\begin{remark}[Kernel Support]
For computational reasons, SPH  typically uses kernel functions with compact support, i.e.,  $K(\bfs, \bfs';h)=0$ if $\| \bfs- \bfs'\| \ge h$ such as the B-spline kernel~\cite{Monaghan92}. However, kernels with infinite support, such as the squared exponential (Gaussian) are also possible.
\end{remark}

\begin{remark}[Isotropy]
Condition (6) is invoked to simplify analytical calculations.  It is also possible to consider kernel functions with  anisotropic distance dependence (geometric anisotropy).  In this case, isotropic dependence can be restored by means of suitable rescaling and rotation transformations, e.g.~\cite{ChoHri08,dth17b}.
\end{remark}

The kernel-based representation used in SPH  and herein has roots in some theoretical results which we review below.

\begin{definition}[Representer of Evaluation]
If $x(\bfs)$ is in
the  Hilbert space $L_{2}$ with the inner product~\eqref{eq:inner-scalar}, the Dirac delta function becomes the representer of evaluation~\cite{Schoelkopf01}:
\beq
\label{eq:dirac-representer}
x(\bfs) = \int_{\Rd} \D{\bfs}' \, \delta(\bfs-\bfs')\, x(\bfs').
\eeq
\end{definition}

If the delta function, which is not in the Hilbert space $L_{2}$ itself, is replaced with a smoother kernel function $K(\cdot, \cdot)$, the field is approximated as
\beq
\label{eq:sph-field}
x(\bfs) \approx  x^{\ast}(\bfs) = \int_{\Rd} \D{\bfs}' \, K(\bfs-\bfs';h)\, x(\bfs'),
\eeq
where  $K(\bfs-\bfs';h)$ is a smooth kernel function and $h>0$ is the \emph{bandwidth parameter} which determines the kernel's range.

\begin{remark}[Zero-bandwidth limit]
At the zero-bandwidth limit $h \to 0$ it holds that $K(\bfs- \bfs';h) \to \delta(\bfs-\bfs')$, where $\delta(\cdot)$ is the \emph{Dirac delta} function, and thus $x^\ast(\bfs) \to x(\bfs)$. However, note that $\delta(\bfs-\bfs')$ is a generalized function and does not admit derivatives in the ordinary sense.
\end{remark}

Equation~\eqref{eq:sph-field} provides an approximate representation of $x(\bfs)$ in terms of a smoothing kernel. However, if $K(\bfs-\bfs';h)$ is a positive-definite kernel, the approximation  becomes exact for functions in the  reproducing kernel Hilbert space (RKHS)~\cite{Smola98} by means of the Moore-Aronszajn theorem~\cite{Aronszajn50,Rasmussen06}.  The latter is not used herein.

%

\vspace{6pt}

\begin{definition}[SPH Representation]
\label{defi:sph-pd}
In the SPH representation, a continuous  function $x(\bfs)$ is typically known in terms of a finite  number of ``particles'' $\{ x_{n} \}_{n=1}^{N}$ that represent observations of $x(\bfs)$.
\begin{enumerate}
    \item Let $\bfx = (x_{1}, \ldots x_{N})^{\top}$ represent a vector of sample values collected at the locations $\Samp = \{ \bfs_{n}\}_{n=1}^{N}$, where $\bfs_{n} \in \Samp$ is the position of the $n$-th particle.

    \item Let $\{ v_{n} \}_{n=1}^{N}$, where $\; v_{n}>0, \; \forall n=1, \ldots, N$ represent a vector of positive weights.

    \item Based on the above, define the SPH representation of $x(\bfs)$ by means of
    \beq
    \label{eq:x-sph}
    \xk(\bfs) = \sum_{n=1}^{N} K(\bfs-\bfs_{n};h) \, x_{n}\, v_{n}, \; \mbox{for all} \; \bfs \in \Rd,
    \eeq
 where $K(\cdot, \cdot; h)$ is a kernel function which satisfies the conditions of Definition~\ref{defi:kernel}.
\end{enumerate}

\end{definition}

\begin{remark}[SPH weights]
The SPH weights are typically set to $v_{n} =m_{n}/\rho_{n}$ where $m_n$ represents the mass and $\rho_n$ the mass density of the particle $n$.  If we assume  that each particle occupies the locus of all points
$\bfs^{\ast}_{n}$ such that $\lVert \bfs_{n} - \bfs^{\ast}_{n} \rVert < \lVert \bfs_{m} - \bfs^{\ast}_{n}\rVert$ for all $m \neq n = 1, \ldots, N$, then $v_n$ is given by  the  hyper-volume of the Voronoi cell which is centered at $\bfs_{n}$ in the $d$\textsuperscript{th}-order Voronoi tesselation of $\Samp$.
\end{remark}

In the SPH representation~\eqref{eq:x-sph} it is straightforward to calculate the partial derivatives of $\xk(\bfs)$. Let $\bmal=(\alpha_{1}, \ldots, \alpha_{d})$ represent a multiindex of order $\lvert \bmal \rvert = \alpha_{1} + \ldots + \alpha_{d}$, where the $\alpha_{i} \in \Na$.   The partial derivatives of order $\lvert \bmal \rvert$ of $\xk(\bfs)$ can be evaluated  without requiring a discretization mesh as follows:
\beq
\label{eq:kernel-deriv}
\frac{\pa^{\lvert \bmal \rvert} \xk(\bfs)}{\pa s_{1}^{\alpha_1} \ldots \pa s_{d}^{\alpha_d}} =
\sum_{n=1}^{N} \frac{\pa^{\lvert \bmal \rvert} K(\bfs-\bfs_{n};h)}{\pa s_{1}^{\alpha_1} \ldots \pa s_{d}^{\alpha_d}}  \, x_{n}\, v_{n},
\eeq
where $K(\bfs-\bfs_{n};h)$ admits partial derivatives of order $\alpha$ at least.
Based on the above,  the \emph{SPH Laplacian and Bi-Laplacian derivatives} of $\xk(\bfs)$ are given by
\beq
\label{eq:kernel-deriv-L-B}
\triangle^{\ell} \xk(\bfs) = \sum_{n=1}^{N} \triangle^{\ell} K(\bfs-\bfs_{n};h) \, x_{n}\, v_{n},\quad \ell=1, 2\, .
\eeq
Below we define an energy functional with  interactions  between latent fields $\xk(\bfs)$, which are obtained by means of the discrete convolutions~\eqref{eq:x-sph} representing smoothed particles.

\medskip

\begin{theorem}[SPH-LAP2 Random Field]
\label{theorem:sph-lap2-disc}
Let $\bfx \in \R^{N}$ represent a vector of $N$ field values at the positions of the sample set $\Samp$. The following SPH-LAP2 quadratic functional is defined:
\beq
\label{eq:Ha-fgc-kernel}
\Ha(\bfx;\bmthe) = \frac{1}{2}\left[ \theta_{0} \langle \xk(\bfs), \xk(\bfs)\rangle  + \theta_{1}  \langle \nabla \xk(\bfs), \nabla \xk(\bfs)\rangle +  \theta_{2}  \langle \nabla^{2} \xk(\bfs), \nabla^{2} \xk(\bfs)\rangle\right],
\eeq
where the SPH representation $\xk(\bfs): \Rd \to \R$ is  introduced in Definition~\ref{defi:sph-pd}, and its partial derivatives up to order $\lvert \bmal \rvert=4$ are defined by means of Eq.~\eqref{eq:kernel-deriv-L-B}.  Let us assume that (i) one of the conditions (C1) and (C2) of Theorem~\ref{theorem:gaussian} hold for the parameter vector $\bmthe$ and   (ii) the absolute value of the Fourier transform $\tilde{K}(\bfk;h)$ of the kernel function $K(\bfr;h)$ decays faster than the algebraic function $\lVert \bfk \rVert^{-(d+4)/2}$ as $\lVert \bfk \rVert \to \infty$. Then, the following statements are true:

\begin{enumerate}\itemsep0.5em

\item[(S1)]
The energy functional defined by  Eq.~\eqref{eq:Ha-fgc-kernel} can be expressed as
\beq
\label{eq:Ha-Q}
\Ha(\bfx;\bmthe) = \frac{1}{2}\, \bfx^\top \, \Prec \, \bfx,
\eeq
where $\Prec$ is a positive definite precision matrix with entries $Q_{n,m}= v_{n}\, Q^{\ast}_{n,m} \, v_{m}$,  $\{ v_{n}\}_{n=1}^{N}$ are the SPH weights, $Q^{\ast}_{n,m}={\Prec}^{\ast}(\bfs_{n}-\bfs_{m};\bmthe,h)$, for all $\bfs_{n}, \bfs_{m} \in \Samp$, and ${\Prec}^{\ast}(\bfr;\bmthe,h)=\ift[\tilde{\Prec}^{\ast}(\bfk;\bmthe,h)]$ is the inverse Fourier transform of the spectral function
\beq
 \label{eq:spectral-function}
  \tilde{\Prec}^{\ast}(\bfk;\bmthe,h) = \lVert \tilde{K}(\bfk;h) \rVert^{2} \, \left( \theta_{0} + \theta_{1}  \lVert \bfk \rVert^{2}   + \theta_{2}  \lVert \bfk \rVert^{4} \right).
 \eeq

\item[(S2)]  The exponential function $\rho: \R^{N} \to \R$ defined by
\beq
\label{eq:pdf}
\rho(\bfx;\bmthe,h) = \frac{\left(\det{\Prec}\right)^{1/2}}{(2\pi)^{N/2}} \,\E^{-\Ha(\bfx;\bmthe)},
\eeq
where $\det{\Prec}$ is the determinant of the matrix $\Prec$, represents a  Gaussian joint probability density function for the vector $\bfx$ with  zero mean and covariance matrix $\bfC= \Prec^{-1}$; the pdf is thus fully determined by the precision matrix $\Prec$.

\end{enumerate}
\end{theorem}

\begin{proof}
First we prove statement (S1). This step requires determining the precision matrix $\Prec$.
The three inner products in the energy functional of Eq.~\eqref{eq:Ha-fgc-kernel} are expressed as follows
\begin{align}
 \langle \nabla^{\ell}\xk(\bfs),\, \nabla^{\ell}\xk(\bfs)\rangle = &
 (-1)^{\ell}\sum_{n=1}^{N}\sum_{m=1}^{N}
 \int_{\Rd}\D\bfs \, K(\bfs- \bfs_{n};h) \, \nabla^{2\ell} K(\bfs- \bfs_{m};h) v_{n} v_{m} x_{n} \, x_{m},
 \nonumber \\
  & \mbox{for} \; l=0, 1, 2.
\end{align}
The equation for $\ell=0$ follows trivially from the inner product definition~\eqref{eq:inner-scalar} (keeping in mind $\Do$ is replaced with $\Rd$) and the SPH representation~\eqref{eq:x-sph}.
The equation for $\ell=1$ follows from the inner product~\eqref{eq:inner-vector},  application of integration by parts, and the SPH Laplacian given by Eq.~\eqref{eq:kernel-deriv-L-B}. The boundary terms at infinity, resulting from integration by parts, vanish because the kernel function tends to zero for distances much larger than the bandwidth. Finally, the equation for $\ell=2$ is obtained from the scalar inner product~\eqref{eq:inner-scalar} and application of integration by parts twice.

In a more compact form, the inner products can be expressed as:
\begin{subequations}
\label{eq:BG-inner-products}
\beq
 \langle \nabla^{\ell}\xk(\bfs),\, \nabla^{\ell}\xk(\bfs)\rangle =
 \sum_{n=1}^{N}\sum_{m=1}^{N} v_{n}\, v_{m} \, x_{n} \, x_{m} \, \Phi_{n,m;\ell}(h),
 \eeq
 where the terms $\Phi_{n,m;\ell}(h)$ represent the following $d$-dimensional integrals
 \beq
 \Phi_{n,m;\ell}(h) = (-1)^{\ell} \int_{\Rd}\D\bfs \, K(\bfs- \bfs_{n};h) \, \nabla^{2\ell} K(\bfs- \bfs_{m};h).
 \eeq
\end{subequations}

As it follows  from Eqs.~\eqref{eq:Ha-fgc-kernel} and~\eqref{eq:BG-inner-products}, the energy functional is given by the double sum
 \beq
 \label{eq:Ha-sph-lap2-disc}
 \Ha(\bfx;\bmthe,h) = \frac{1}{2}
 \sum_{n=1}^{N}\sum_{m=1}^{N} v_{n}\, v_{m} \, x_{n} \, x_{m} \,\left[ \theta_{0}   \Phi_{n,m;0}(h)+ \theta_{1}   \Phi_{n,m;1}(h) + \theta_{2}   \Phi_{n,m;2}(h)  \right].
 \eeq

Thus, $\Ha(\bfx;\bmthe,h)$ is expressed in the form of Eq.~\eqref{eq:Ha-Q}, with $\Prec$  defined as
 \beq
 \label{eq:Qnm}
 {Q}_{n,m} \doteq [\Prec]_{n,m} = v_{n} \,\left[ \theta_{0}   \Phi_{n,m;0}(h)+ \theta_{1}   \Phi_{n,m;1}(h) + \theta_{2}   \Phi_{n,m;2}(h)  \right]\, v_{m}.
 \eeq

 \medskip

Furthermore, let us define the  matrix ${\Prec}^\ast$ obtained from $\Prec$ if the SPH weights $\{ v_{n} \}$ are  set equal to unity. Then, from Eq.~\eqref{eq:Qnm} it follows that
\beq
 \label{eq:Q-tilde}
 {\Prec}^\ast = \theta_{0}   {\bm \Phi}_{0}(h)+ \theta_{1}   {\bm \Phi}_{1}(h) + \theta_{2}   {\bm \Phi}_{2}(h),
 \eeq
 where $[{\bm \Phi}_{\ell}(h)]_{n,m}=\Phi_{n,m;\ell}(h)$, $\ell=0,1,2$. \emph{If ${\Prec}^\ast $ is positive-definite then so is $\Prec$}. This follows from the fact that if ${\Prec}^\ast $ is positive definite, then $\bfx^\top {\Prec}^{\ast} \bfx \ge 0$ for all vectors $\bfx \in \R^{N}$.  Therefore, the inequality also holds if $\bfx$ is replaced with $\bfx'$ such that $x'_{n} = v_{n} x_{n}$, $\forall \, n =1, \ldots, N$.

Expressing the kernel functions in terms of the respective IFTs, i.e., Eq.~\eqref{eq:invcovft}, the $\Phi_{n,m;\ell}(h)$  are given by the following spectral integrals~\footnote{In the proof, to abbreviate the notation, we drop the dependence of $\tilde{K}(\bfk;h)$ on $h$.}

 \begin{align}
 \label{eq:Phi-nm-ell}
\Phi_{n,m;\ell}(h) = & \frac{(-1)^{\ell}}{(2\pi)^{2d}}   \int_{\Rd}\D\bfs \, \int_{\Rd}\D\bfk_{1}\, \E^{\jmath \bfk_{1}\cdot (\bfs - \bfs_{n})} \tilde{K}(\bfk_{1})\,
\int_{\Rd}\D\bfk_{2}\, \nabla^{2\ell} \,\E^{\jmath \cdot \bfk_{2}(\bfs - \bfs_{m})} \tilde{K}(\bfk_{2})
\nonumber \\
 = & \frac{1}{(2\pi)^{2d}} \int_{\Rd}\D\bfk_{1}\,\int_{\Rd}\D\bfk_{2}\, \left[ \int_{\Rd}\D\bfs \,  \E^{\jmath (\bfk_{1}+ \bfk_{2}) \cdot \bfs} \right] \,  \E^{-\jmath (\bfk_{1}\cdot \bfs_{n} + \bfk_{2} \cdot \bfs_{m}) }\, \tilde{K}(\bfk_{1})\, \lVert \bfk_{2} \rVert^{2\ell} \tilde{K}(\bfk_{2})
 \nonumber \\
 = & \frac{(2\pi)^{d}}{(2\pi)^{2d}}\int_{\Rd}\D\bfk_{1}\,\int_{\Rd}\D\bfk_{2}\, \delta(\bfk_{1} + \bfk_{2})\,  \E^{-\jmath (\bfk_{1}\cdot \bfs_{n} + \bfk_{2} \cdot \bfs_{m}) } \, \tilde{K}(\bfk_{1})\, \lVert \bfk_{2} \rVert^{2\ell}  \tilde{K}(\bfk_{2})
 \nonumber \\
 = & \frac{1}{(2\pi)^{d}}\int_{\Rd}\D\bfk_{1}\,\,  \E^{-\jmath \bfk_{1}\cdot ( \bfs_{n} - \bfs_{m}) } \, \lVert \bfk_{1} \rVert^{2\ell}  \tilde{K}(\bfk_{1})\, \tilde{K}(-\bfk_{1})
 \nonumber \\
 = & \frac{1}{(2\pi)^{d}}\int_{\Rd}\D\bfk\,\,  \E^{\jmath \bfk\cdot ( \bfs_{n} - \bfs_{m}) } \, \lVert \bfk \rVert^{2\ell}  \, \lVert \tilde{K}(\bfk) \rVert^{2} \, .
 \end{align}

The second line in Eq.~\eqref{eq:Phi-nm-ell} is obtained using  ${\mathrm {FT}}[\nabla^{2}]= -\lVert \bfk \rVert^2$, i.e., that the image of the Laplace operator in Fourier space is $-\lVert \bfk \rVert^2$~\cite{Trefethen05}.  The third line follows from the orthonormality of the Fourier basis, i.e., $\int_{\Rd} \D\bfs \, \exp(\jmath \bfk \cdot \bfs) = (2\pi)^{d}\, \delta(\bfk)$. The fourth line is obtained by integrating the delta function, i.e., $\int_{\Rd} \D \bfk_{2} \,\delta(\bfk_1 + \bfk_2) \, \phi(\bfk_{2})= \phi(-\bfk_{1})$. Finally, the fifth line is obtained by renaming $\bfk_{1} \to -\bfk$ and using the conjugation property for the Fourier transform of real-valued functions, i.e.,  $\tilde{K}(-\bfk)= \tilde{K}^{\dag}(\bfk)$, where the symbol $\dagger$ denotes the complex conjugate.

 Based on Eqs.~\eqref{eq:Q-tilde} and~\eqref{eq:Phi-nm-ell} the entries of the matrix ${\Prec}^{\ast}$ are given by the following spectral integral
 \beq
 \label{eq:precision-function}
 {Q}^{\ast}_{n,m} = \frac{1}{(2\pi)^{d}}\int_{\Rd}\D\bfk\,\,  \E^{\jmath \bfk\cdot ( \bfs_{n} - \bfs_{m}) } \,  \lVert \tilde{K}(\bfk) \rVert^{2} \, \left[ \theta_{0} + \theta_{1}  \lVert \bfk \rVert^{2}
  + \theta_{2}  \lVert \bfk \rVert^{4} \right].
 \eeq
 Therefore, ${\Prec}^\ast $ is given by the inverse Fourier transform of the \emph{SPH-LAP2 spectral function} of Eq.~\eqref{eq:spectral-function}.

 According to Bochner's theorem, the $N \times N$ matrix $\Prec^\ast$ is positive definite for any $N \in \Na$ if and only if (A1) $\tilde{\Prec}^{\ast}(\bfk;\bmthe,h) \ge 0$ and (A2) the improper integral
 $\int_{\Rd} \D\bfk \,\tilde{\Prec}^{\ast}(\bfk;\bmthe,h)$ is finite. The condition (A1) is satisfied by (C1) and (C2) of Theorem~\ref{theorem:gaussian}. The condition (A2) holds iff
 $\lim_{\kk \to \infty} \lVert \tilde{K}(\bfk) \rVert \, \lVert \bfk \rVert^{d/2 +2} = 0$ so that the improper integral exists as $\lVert \bfk \rVert \to \infty$.  This concludes the proof of Statement (S1).

 \medskip
 Statement (S2) follows straightforwardly  from Eq.~\eqref{eq:Ha-Q} given the positive-definiteness of the precision matrix $\Prec$ and the definition of the multivariate normal distribution with covariance matrix $\bfC = \Prec^{-1}$, e.g.~\cite{Rasmussen06}.
\end{proof}

\medskip

\begin{remark}[Spatial coupling  function]
If  we replace $\bfs_{n} - \bfs_{m}$ with $\bfr$ in Eq.~\eqref{eq:precision-function} we obtain a \emph{radial, positive-definite function} $\Prec^{\ast}(r;\bmthe,h)$ which can be used to construct a mesh-free precision matrix for any sampling point configuration $\Samp$.  To our knowledge, this is the first such construction.
\end{remark}

\begin{remark}[On the homogeneity of the precision matrix]
If $v_{n}=1$ for all $n=1, \ldots, N$ then $\Prec= \Prec^{\ast}$, and $\Prec$ is a position-independent precision matrix which depends only on the spatial lag.
On the other hand, if there exists at least one $j \in \{1, 2, \ldots, N\}$ such that $v_{j} \neq 1$, then $[\Prec]_{n,m}= v_{n} \, v_{m}\,[\Prec^{\ast}]_{n,m}$, and $\Prec$ is a spatially non-homogeneous precision matrix. In the following, we focus on the precision matrix $\Prec^\ast$ assuming that $v_{n}=1$ for all $n=1, \ldots, N$.
\end{remark}

\begin{remark}[On the connection with Gauss-Markov random fields] Let us consider that $v_{n}=1$ for all $n=1, \ldots, N$. This implies an SPH-LAP2 random field with a homogeneous precision matrix.
The exponential joint probability density of Eq.~\eqref{eq:pdf} has the same formal expression as the pdf of a Gauss-Markov random field (GMRF) with precision matrix $\Prec$; e.g., see Eq.~(2.4) in \cite{Rue05}.
According to Rozanov a stationary field is a Markov random field if and only if the spectral density is a reciprocal of a polynomial function~\cite{Rozanov77,Rozanov82}.  This  implies a polynomial generalized Fourier transform of the inverse covariance. This condition is clearly not satisfied by the SPH-LAP2 precision spectral  function~\eqref{eq:spectral-function} unless
$\lVert \tilde{K}(\bfk;h) \rVert = 1$ which can only be realized for $h=0$. The SPH-LAP2 random field fails to satisfy the Rozanov criterion because it is non-stationary for $h \neq 0$ (see Remark~\ref{rem:nonsta}).
\end{remark}

\subsection{The SPH-LAP2 Precision Functions}
\label{ssec:sph-lap2-prec}

In this section we derive a general expression for the SPH-LAP2 spatial coupling (precision) function which is valid for all smoothing kernels that satisfy the conditions of Definition~\ref{defi:kernel}. In order to accomplish this goal we need to  evaluate the Laplacian and Bi-Laplacian of radial functions.

If $C(r): \R \to \R$   where $C(r) = C_{d}(\rr): \Rd \to \R$ is a radial function which admits at least second-order continuous partial derivatives with respect to $r= \rr$, its Laplacian  is given by~\cite[p.~16]{Yadrenko83},\cite[p.~446]{dth20}
\begin{align}
\label{eq:laplacian-radial}
\triangle C(r) = & \frac{1}{r^{d-1}}\, \frac{\D}{\D r}\left[  r^{d-1}\frac{\D C(r)}{\D r}\right] = \frac{\D^{2}C(r)}{\D r^{2}} + \frac{(d-1)}{r} \frac{\D C(r)}{\D r}.
\end{align}

\medskip

\begin{lemma}[Bi-Laplacian of Radial function]
\label{lemma:bilaplacian}
Let $C(r): \R \to \R$, where $r=\rr$ be a radial function  that admits at least fourth-order continuous derivatives with respect to $r$.  The Bi-Laplacian of $C(r)$ is given by
\begin{align}
\label{eq:bilaplacian-radial}
\triangle^{2} C(r)  =  & \frac{\D^{4}C(r)}{\D r^{4}} + \frac{2(d-1)}{r}\frac{\D^{3}C(r)}{\D r^{3}}
 + \left[ \frac{(d-1)^2 - 2(d-1) }{ r^{2}} \right]\, \left[ \frac{\D^{2} C(r)}{\D r^{2}} - \frac{1}{r} \frac{\D C(r)}{\D r} \right].
\end{align}
\end{lemma}

\medskip
\begin{proof}
It holds that $\triangle^{2} C(r) = \triangle [\triangle C(r)]$, where $\triangle C(r)$ is given by Eq.~\eqref{eq:laplacian-radial}. Let $\triangle C(r)=c_{1}(r) + c_{2}(r)$, where $c_{1}(r)=\D^{2}C(r) / \D r^{2}$  and $c_{2}(r)= \frac{d-1}{r} \D C(r) / \D r$. Then $\triangle^{2} C(r)=\triangle c_{1}(r) + \triangle c_{2}(r)$. We can calculate $\triangle c_{i}(r)$, where $i=1, 2$ by application of Eq.~\eqref{eq:laplacian-radial}.
For $i=1$ we obtain
\beq
\label{eq:bil-c1}
\triangle c_{1}(r) = \frac{\D^{2} c_{1}(r)}{\D r^{2}} + \frac{(d-1)}{r} \frac{\D c_{1}(r)}{\D r} =
\frac{\D^{4} C(r)}{\D r^{4}} + \frac{(d-1)}{r} \frac{\D^{3} C(r)}{\D r^3}.
\eeq
Similarly, for $i=2$ we obtain
\[
\triangle c_{2}(r) = \frac{\D^{2}}{\D r^{2}}  \left[\left(\frac{d-1}{r} \right) \frac{\D C(r) }{\D r}\right]  + \frac{(d-1)}{r} \frac{\D}{\D r} \left[\left(\frac{d-1}{r} \right)\frac{\D C(r) }{\D r}. \right]
\]
Let us define $g_{1}(r) \doteq \frac{d-1}{r}$ and $g_{2}(r) \doteq \frac{\D C(r)}{\D r}$. Then,  $\triangle c_{2}(r)$ is given by
\[
\triangle c_{2}(r) = \frac{\D^{2}}{\D r^{2}} \left[ g_{1}(r) g_{2}(r) \right] + \left( \frac{d-1}{r} \right)\frac{\D}{\D r} \left[ g_{1}(r) g_{2}(r) \right].
\]
The first and second order derivatives of the product $g_{1}(r) g_{2}(r)$ can be evaluated using differentiation by parts leading to
\[
\frac{\D }{\D r}\left[ g_{1}(r) g_{2}(r) \right]= g_{1}(r)\frac{\D  g_{2}(r)}{\D r}+ g_{2}(r)\frac{\D  g_{1}(r)}{\D r},
\]
\[
\frac{\D^{2} }{\D r^{2}}\left[ g_{1}(r) g_{2}(r) \right]= g_{1}(r)\frac{\D^{2}  g_{2}(r)}{\D r^{2}}+ g_{2}(r)\frac{\D^{2}  g_{1}(r)}{\D r^{2}} + 2 \frac{\D  g_{1}(r)}{\D r}\frac{\D  g_{2}(r)}{\D r}.
\]
Based on the above, we obtain the following equation for $\triangle c_{2}(r)$
\begin{align}
\label{eq:bil-c2}
 \triangle c_{2}(r) = & \frac{2(d-1)}{r^{3}}\frac{\D C(r)}{\D r} + \frac{(d-1)}{r}\frac{\D^3 C(r)}{\D r^3} - \frac{2(d-1)}{r^{2}}\frac{\D^2 C(r)}{\D r^2}
 \nonumber \\
& + \frac{(d-1)}{r}\left[ \frac{(d-1)}{r} \frac{\D^2 C(r)}{\D r^2} - \frac{(d-1)}{r^2}\frac{\D C(r)}{\D r}\right].
\end{align}
Equation~\eqref{eq:bilaplacian-radial} follows by adding the respective terms on each side of  Eqs.~\eqref{eq:bil-c1}-\eqref{eq:bil-c2}.
\end{proof}

\medskip

\begin{definition}[Spatial interaction]
Let  $\bfx$ represent a vector of sample values measured at the positions in the set $\Samp$, distributed according to the Boltzmann-Gibbs joint density $\ff \propto \E^{-\Ha[\bfx]}$ where $\Ha[\bfx]$ is the energy functional of Eq.~\eqref{eq:Ha-Q}.  (i) Two points $\bfs_{n}, \bfs_{m} \in \Samp$ are  considered to interact in the strict sense if and only if $Q_{n,m} \neq 0$. (ii) If on the other hand $Q_{n,m} =0$, the values of the field at $\bfs_{n}$ and $\bfs_{m}$ are conditionally independent. (iii)
Two points are considered to  interact  in the weak sense if and only if
$\lvert Q_{n,m} \rvert < \epsilon$, where $0< \epsilon \ll 1$ is a small, user-defined threshold.
\end{definition}

The strong definition of spatial interaction is useful if the smoothing kernel $K(\cdot,\cdot;h)$ is compactly supported. Then, sparse precision matrices $\Prec$ can be obtained.  For domains whose length along orthogonal directions is approximately uniform,  sparsity requires that $h \ll L$ where $L$ is a characteristic domain size. If the domain length varies along orthogonal directions, sparsity requires that  $h \ll \overline{L}$, where $\overline{L}$ is an ``average'' domain size.  The bandwidth $h$ is a parameter to be determined from the data. Experience with kernel methods shows that sparse precision matrices are typically obtained if the data do not have long-range correlations~\cite{dth15,dth19}.
Infinitely extended smoothing kernels (e.g., exponential, squared exponential)  lead to dense precision matrices $\Prec$.  The latter can be approximated by a sparse  $\Prec'$ if   weak interactions (with respect to an arbitrary small threshold $\epsilon$) are truncated.  Conditional independence can be understood in the SPH framework as follows: the interactions are defined by the latent field $\xk(\bfs)$ via the Laplacian and Bi-Laplacian operators.  The latent fields are thus coupled locally via the differential operators (or their discretized approximations). At the same time, if the smoothing kernel is compactly supported, the latent fields couple only points in their vicinity. Thus, the  interactions in the SPH representation are restricted only over a maximum range $r_{\max}$ (determined by the bandwidth, the LAP2 operator, and the kernel function), while $Q_{n,m}$ vanishes if  $\lVert \bfs_{n} - \bfs_{m} \rVert > r_{\max}$.

\begin{definition}[Spatial interaction smoothing function]
\label{defi:K2}
The  function $K^{(2)}(r;h)$ is defined as the inverse Fourier transform of the squared modulus of the kernel function $K(\bfr;h)$, i.e.,  $ K^{(2)}(r;h) =\ift\left[ \lVert \tilde{K}(\bfk;h) \rVert^{2}\right]$, where $r=\rr$.
\end{definition}

As will be shown below, the function $K^{(2)}(r;h)$  and its derivatives determine the  dependence of interactions on the  distance between  points. Therefore $K^{(2)}(r;h)$ is a main ingredient of the SPH-LAP2 precision function.  Since the spatial  interactions, implemented via Eq.~\eqref{eq:precision-function} are well-defined for all pairs of points, the spatial coupling function based on $K^{(2)}(r;h)$  bypasses the need to define a specific mesh (spatial grid).

\medskip

\begin{theorem}[SPH-LAP2 Precision Function]
\label{theorem:sph-lap2-prec}
Let the SPH-LAP2 precision function be defined by means of the $d$-dimensional inverse Fourier transform
\beq
\label{eq:precision-function-ift}
Q^{\ast}(\bfr;h) = \ift\left[\,  \lVert \tilde{K}(\bfk;h) \rVert^{2} \, \left( \theta_{0} + \theta_{1}  \lVert \bfk \rVert^{2}   + \theta_{2}  \lVert \bfk \rVert^{4} \right) \,\right],
\eeq
where the parameter vector $\bmthe$ satisfies the conditions of Theorem~\ref{theorem:gaussian}, and $\tilde{K}(\bfk)$ is the FT of a radial kernel function $K(r)$ (see Definition~\ref{defi:kernel}). Then, $Q^{\ast}(\bfr;\bmthe,h)$ is  given by
\beq
\label{eq:precision-sph}
Q^{\ast}(r;h) =  \theta_{0}\,K^{(2)}(r;h) - \theta_{1} \, \triangle K^{(2)}(r;h)
+ \theta_{2} \triangle^2 K^{(2)}(r;h),
\eeq
where $K^{(2)}(r;h)$ is the spatial interaction smoothing function  of  Definition~\ref{defi:K2}. More specifically, it holds that
\begin{align}
\label{eq:precision-sph-2}
 Q^{\ast}(r;h) = & \theta_{0}\,K^{(2)}(r;h) - \theta_{1} \,\left[
 \frac{\D^{2}K^{(2)}(r;h)}{\D r^{2}} + \frac{(d-1)}{r} \frac{\D K^{(2)}(r;h)}{\D r} \right]
 \nonumber \\
 & \quad \quad + \theta_{2} \left\{ \frac{\D^{4}K^{(2)}(r;h)}{\D r^{4}} +
  \frac{2(d-1)}{r}\frac{\D^{3}K^{(2)}(r;h)}{\D r^{3}}  \right.
 \nonumber \\
 &  \left.  \quad \quad \quad
  + \left[ \frac{(d-1)^2 - 2(d-1) }{ r^{2}} \right]\, \left[ \frac{\D^{2} K^{(2)}(r;h)}{\D r^{2}} - \frac{1}{r} \frac{\D K^{(2)}(r;h)}{\D r} \right]\right\} \,.
\end{align}
\end{theorem}

\begin{proof}
Since $K(r;h)$ is a radial function, the FT $\tilde{K}(\bfk;h)$ depends purely on the Euclidean norm of $\bfk$. Hence, the inverse Fourier transform of $\lVert \tilde{K}(\bfk;h) \rVert^{2}$ is also a radial function which depends on the Euclidean norm $r=\rr$ of the lag vector $\bfr$.

Since the image of the Laplace operator  in the Fourier domain is $\ft[\triangle]=-\kk^2$, the  inverse Fourier transforms of $\kk^{2n} \, \lVert \tilde{K}(\bfk;h) \rVert^{2}$, $n=1, 2$, are given by
\begin{align*}
\ift[\,\kk^{2} \, \lVert \tilde{K}(\bfk;h) \rVert^{2}\,]  & = - \triangle K^{(2)}(r;h),
\\[1.5ex]
\ift[\,\kk^{4} \, \lVert \tilde{K}(\bfk;h) \rVert^{2}] & =  \triangle^2 K^{(2)}(r;h) \,.
\end{align*}
Equation~\eqref{eq:precision-sph} follows from the above and the IFT given by  Eq.~\eqref{eq:precision-function-ift}. Finally, Equation~\eqref{eq:precision-sph-2} follows from Eq.~\eqref{eq:precision-sph}, the Laplacian, Eq.~\eqref{eq:laplacian-radial}, and the Bi-Laplacian of a radial function,  Eq.~\eqref{eq:bilaplacian-radial}
of Lemma~\ref{lemma:bilaplacian}.
\end{proof}

\subsection{The Case of Gaussian Smoothing Kernel}
\label{ssec:sph-lap2-prec-gauss}

In this section we assume that the smoothing kernel of Definition~\ref{defi:kernel} is given by the squared exponential (Gaussian) function with bandwidth $h$:
\beq
\label{eq:square-expo-kernel}
\Kse(\bfr;h) = \frac{1}{\left( h\sqrt{\pi}\right)^d}\exp\left[ - \lVert \bfr \rVert^2/h^2 \right].
\eeq

\begin{lemma}[Gaussian spatial interaction smoothing function]
\label{lemma:K2}
Let $\Kse(\bfr;h)$, where $\bfr \in \Rd$, be the squared exponential kernel function with Fourier transform $\Kseft(\bfk h)$; also let $r=\rr$. The spatial interaction smoothing function   $K^{(2)}(r;h) =\ift\left[ \lVert \Kseft(\bfk h) \rVert^{2}\right]$,  is given by
\beq
\label{eq:K2}
K^{(2)}(r;h) =  \Kse(\bfr;h') = \frac{1}{\left(h\,\sqrt{2\pi}\right)^{d}} \,\E^{- \lVert \bfr \rVert^2/2 h^2}, \; \mbox{where} \; h'=\sqrt{2}h.
\eeq

\end{lemma}

\begin{proof}
The Fourier transform of the Gaussian function $\exp(-\rr^2/h^2)$ is also a Gaussian given by~\cite[p.~160]{dth20}
\[
\mbox{FT}\left[ \E^{-\rr^2/h^2}\right] =  \pi^{d/2}\, h^{d}\,
\exp\left( - k^2 h^2/4\right), \; \mbox{where} \; k=\kk.
\]
Hence, the Fourier transform of the smoothing kernel of Eq.~\eqref{eq:square-expo-kernel} is given by the Gaussian function
\beq
\label{eq:fft-gauss}
\Kseft(k h) = \exp\left( - k^2 h^2/4\right).
\eeq
Thus, the modulus squared of $\Kseft(k h)$ is also a Gaussian function given by
\beq
\label{eq:fft-gauss-square}
\lVert \Kseft(k h) \rVert^{2} = \E^{- k^2 h^2/2} =
\Kseft(k h'), \; \mbox{where} \;  h' = \sqrt{2} h.
\eeq
Based on Eq.~\eqref{eq:fft-gauss-square}, the inverse Fourier transform of $K^{(2)}(r;h)$ is given by Eq.~\eqref{eq:K2}.
\end{proof}

\medskip

\begin{theorem}[SPH-LAP2 precision function based on Gaussian smoothing kernel]
\label{theorem:sph-lap2-prec-gauss}
Let us consider the SPH-LAP2 precision function $Q^{\ast}(\bfr;\bmthe,h)$  defined by means of Eq.~\eqref{eq:precision-function-ift} and equipped with the squared-exponential kernel.  Then, the precision function in the spectral domain is given by
\beq
\label{eq:Qk-gauss}
\tilde{Q}^{\ast}(k;\bmthe,h) =   \E^{-k^2h^2/2}\, \left( \theta_{0} + \theta_{1}  k^{2}   + \theta_{2}  k^{4} \right), \; \mbox{where} \; k= \kk.
\eeq
The precision function $Q^{\ast}(\bfr;\bmthe,h)$  is given by the following radial function
\begin{align}
\label{eq:precision-squared-exponential}
Q^{\ast}(r;\bmthe,h)  = \frac{\E^{-r^2/2h^2}}{\left(h\,\sqrt{2\pi}\right)^{d}} \, & \left\{ \theta_{0}\, -  \,\frac{\theta_{1}}{h^{2}}\, \left(  \frac{r^2}{h^2} - d\right) +
\frac{\theta_{3}}{h^4}\left[\frac{r^4}{h^4} - 2(d+2) \frac{r^2}{h^2} + d(d + 2) \right] \right\}.
\end{align}

\end{theorem}

\medskip

\begin{proof}

The precision spectral function  is obtained from the IFT of  Eq.~\eqref{eq:precision-function-ift} and the Fourier transform~\eqref{eq:fft-gauss-square} of the squared exponential kernel. All the factors involved are radial functions of $k=\kk$. Hence, Eq.~\eqref{eq:Qk-gauss} is proved.

As is evident in Eq.~\eqref{eq:Qk-gauss},  $\tilde{Q}^{\ast}(k;\bmthe,h)$ decays exponentially with $k$ and  lacks singularities on the real line; this  ensures that $\tilde{Q}^{\ast}(k;\bmthe,h)$  is absolutely integrable so that its inverse Fourier transform exists. Furthermore, since the parameter vector $\bmthe$ satisfies the conditions of Theorem~\ref{theorem:gaussian}, based on Bochner's theorem the inverse Fourier transform of Eq.~\eqref{eq:Qk-gauss} is a positive definite  function $Q^{\ast}(\bfr;\bmthe,h)$.

Based on Eq.~\eqref{eq:precision-function-ift}, $Q^{\ast}(\bfr;\bmthe,h)$ is obtained from the following superposition of inverse Fourier transforms
\[
Q^{\ast}(\bfr;\bmthe,h) =   \theta_{0}\,\ift[\Kseft(k h')] + \theta_{1} \,  \ift[k^{2} \, \Kseft(k h')]
+ \theta_{2}\,  \ift[k^{4} \,\Kseft(k h')].
\]
Since the image of the Laplace operator in the Fourier domain is $-k^2$, the second and third inverse Fourier transforms in the above equation are evaluated as follows:
\[
\ift[\,k^{2} \, \Kseft(k h')\,] = - \triangle \Kse(\bfr; h'), \quad
\ift[\,k^{4} \, \Kseft(k h')] =  \triangle^2 \Kse(\bfr; h') \,.
\]
Taking into account the definition of the Gaussian spatial interaction smoothing function in Eq.~\eqref{eq:K2}, the SPH-LAP2 precision function becomes
\beq
\label{eq:precision-function-gauss-1}
Q^{\ast}(\bfr;\bmthe,h) = \left[ \theta_{0}\,K^{(2)}(r;h)  - \theta_{1} \, \triangle K^{(2)}(r;h)
+ \theta_{2} \triangle^2 K^{(2)}(r;h)  \right].
\eeq
The Laplacian and Bi-Laplacian of the radial function $K^{(2)}(r;h)$  involves up to fourth-order derivatives of the Gaussian function $\exp(-r^2/2h^2)$,   according to Eqs.~\eqref{eq:laplacian-radial} and~\eqref{eq:bilaplacian-radial}.
These derivatives can be evaluated by means of the \emph{Hermite polynomials} $\He_{n}(x): \, \R \to \R$ of order $n\in \{1,2,3,4\}$,  which are defined  as follows:
\beq
\label{eq:hermite}
\He_{n}(x)= (-1)^n \exp(x^2/2) \, \frac{ \D^{n}\exp(-x^2/2) }{\D x^n}.
\eeq
In particular, the Hermite polynomials for $n=1,2,3,4$ are given by
\beq
\label{eq:hermite-1-4}
\He_{1}(x)=x, \; \He_{2}(x)=x^2-1, \; \He_{3}(x)=x^{3} - 3x, \; \He_{4}(x)=x^{4} - 6x^{2}+3\, .
\eeq

According to Eq.~\eqref{eq:hermite}, the $n$-th order derivative of the Gaussian  $\Kse(r; h)=\exp(-r^2/2h^2)$ is given by means of the Hermite polynomials, $\He_{n}(x)$,  as follows
\beq
\label{eq:derivative-gauss}
\frac{\D^{n} \E^{-r^2/2h^2}}{\D r^n} =
\frac{(-1)^{n}}{ h^n}\, \He_{n}\left(\frac{r}{ h}\right)\,  \E^{-r^2/2h^2}, \quad n \in \Na.
\eeq

Then, based on Eq.~\eqref{eq:laplacian-radial} and Eq.~\eqref{eq:derivative-gauss},  the Laplacian of $K^{(2)}(r;h) $ is given by
\beq
\label{eq:laplacian-kse-1}
 \triangle K^{(2)}(r;h)  = \frac{\E^{-r^2/2h^2}}{\left(h\,\sqrt{2\pi}\right)^{d}} \, \left[ \frac{1}{ h^{2}}\, \He_{2}\left(\frac{r}{ h}\right) +
\left( \frac{1-d}{r h} \right) \,\He_{1}\left(\frac{r}{ h}\right) \right]\, .
\eeq
Using the explicit expressions for $\He_{2}(\cdot)$ and $\He_{1}(\cdot)$ from Eq.~\eqref{eq:hermite}, $\triangle K^{(2)}(r;h)$ becomes
\begin{align}
\label{eq:laplacian-kse-2}
 \triangle K^{(2)}(r;h)  = & \frac{\E^{-r^2/2h^2}}{\left(h\,\sqrt{2\pi}\right)^{d}} \, \left[ \frac{1}{ h^{2}}\, \left(\frac{r^2}{h^2} -1 \right)+
 \frac{1-d}{h^2}  \right]
 \nonumber \\
 = &
 \frac{\E^{-r^2/2h^2}}{\left(h\,\sqrt{2\pi}\right)^{d}} \, \left[ \frac{1}{ h^{2}}\, \left(\frac{r^2}{h^2} -d \right) \right] \,.
\end{align}

Similarly, based on Eqs.~\eqref{eq:bilaplacian-radial} and~\eqref{eq:derivative-gauss} the Bi-Laplacian of $\Kse(r; h')$ is given by
\begin{align}
\label{eq:bilaplacian-kse}
 \triangle^{2} K^{(2)}(r;h) =   \frac{\E^{-r^2/2h^2}}{\left(h\,\sqrt{2\pi}\right)^{d}} \, & \left[ \frac{1}{h^4}\He_{4}\left(\frac{r}{h}\right) +
 \frac{2(1-d)}{rh^3} \,\He_{3}\left(\frac{r}{h}\right)
\right.
 \nonumber \\
  & \quad \left.
  +\frac{(d-1)^2 -2(d-1)}{(rh)^{2}} \,\He_{2}\left(\frac{r}{h}\right)\right.
\nonumber \\
  & \quad \left.  +  \frac{(d-1)^2 - 2(d-1)}{r^{3} h}  \He_{1}\left(\frac{r}{h}\right)\right]\,.
\end{align}

Note that the right hand-side of  Eq.~\eqref{eq:bilaplacian-kse} includes two singular terms $\Or(r^{-2})$ coming from the Hermite polynomials $\He_{2}(x)$ and $\He_{1}(x)$, cf. Eqs.~\eqref{eq:hermite-1-4}.  However, the singular terms cancel each other. Thus, the sum of  terms proportional to $\He_{2}(x)$ and $\He_{1}(x)$ becomes $h^{-4}\,[(d-1)^2 - 2(d-1)]$.  The remaining terms $\sim \Or(h^{-4})$ are evaluated using the expressions for $\He_{4}(x)$ and $\He_{3}(x)$ from Eqs.~\eqref{eq:hermite-1-4}.

Finally, based on  Eq.~\eqref{eq:laplacian-kse-2} for the Laplacian of $K^{(2)}(r;h)$, Eq.~\eqref{eq:bilaplacian-kse} for the Bi-Laplacian of $K^{(2)}(r;h)$,  and Eq.~\eqref{eq:precision-function-gauss-1} for $Q^{\ast}(r;\bmthe,h)$,  the expression Eq.~\eqref{eq:precision-squared-exponential} is obtained for the SPH-LAP2 precision function
equipped with the Gaussian kernel.
\end{proof}

\medskip

\begin{remark}[Non-stationarity]
\label{rem:nonsta}
The SPH-LAP2 random field based on the Gaussian smoothing kernel is not stationary for $h \neq 0$.  To see this, consider that the Fourier transform, $\tilde{Q}^{\ast}(k;\bmthe,h)$, of the precision function is given by  Eq.~\eqref{eq:Qk-gauss}.
If the SPH-LAP2 random field were stationary,
the inverse of the spectral image of the precision operator, i.e., $1/\tilde{Q}(k;\bmthe,h)$,  would correspond to the field's spectral density.  However, as is evident from Eq.~\eqref{eq:Qk-gauss},  $\tilde{Q}^{-1}(k;\bmthe,h)$ is not integrable over $\Rd$ due to the exponential increase of the Gaussian function with $k$; therefore, it is
not a valid spectral density according to Bochner's theorem~\cite{Bochner59}.

\end{remark}

\begin{remark}[Spatial interaction at zero separation]
Based on Eq.~\eqref{eq:precision-squared-exponential} and the definition of the Hermite polynomials given in Theorem~\ref{theorem:sph-lap2-prec-gauss} it follows that the precision function at zero separation (lag) for $h>0$ is given by
\[
Q^{\ast}(0;\bmthe,h) = \frac{1}{\left(h\,\sqrt{2\pi}\right)^{d}} \, \left[ \theta_0  + \frac{\theta_{1} d}{h^2} + \frac{\theta_{2}(d^2 +2 d)}{h^4}   \right].
\]
It is straightforward to show that $Q^{\ast}(0;\bmthe,h)$ is the global maximum of $Q^{\ast}(r;\bmthe,h)$.  Then,  the above equation implies that for a fixed parameter vector $\bmthe$ the value $Q^{\ast}(0;\bmthe,h)$ is lowered as $h$ increases, marking a weakening of the spatial coupling. On the other hand, since $\bmthe$ is a vector of free parameters, it is possible to choose $\theta_{0} \propto h^d$, $\theta_{1} \propto h^{d+2}$ and $\theta_{2} \propto h^{d+4}$ so that finite values of $Q^{\ast}(0;\bmthe,h)$ are obtained even for $h\to 0$.
\end{remark}

\medskip

Plots of the normalized SPH-LAP2 precision function $Q^{\ast}(r;\bmthe,h)/Q^{\ast}(0;\bmthe,h)$ based on the Gaussian kernel are shown in Fig.~\ref{fig:fgc-sph}.  The curves are obtained from Eq.~\eqref{eq:precision-squared-exponential} for $d=2$ and $h=1.5$ (kernel bandwidth).  Three different parameter sets are used. The curves exhibit negatively-valued valleys over a range of lags. All three curves approach zero as $r \to \infty$, either from the positive or the negative side.  Negative values of the precision function for a given separation favor opposite signs of the field values at this distance, because they tend to lower the energy in Eq.~\eqref{eq:Ha-Q}.


\begin{figure}[tb]
\includegraphics[width=0.85\textwidth]{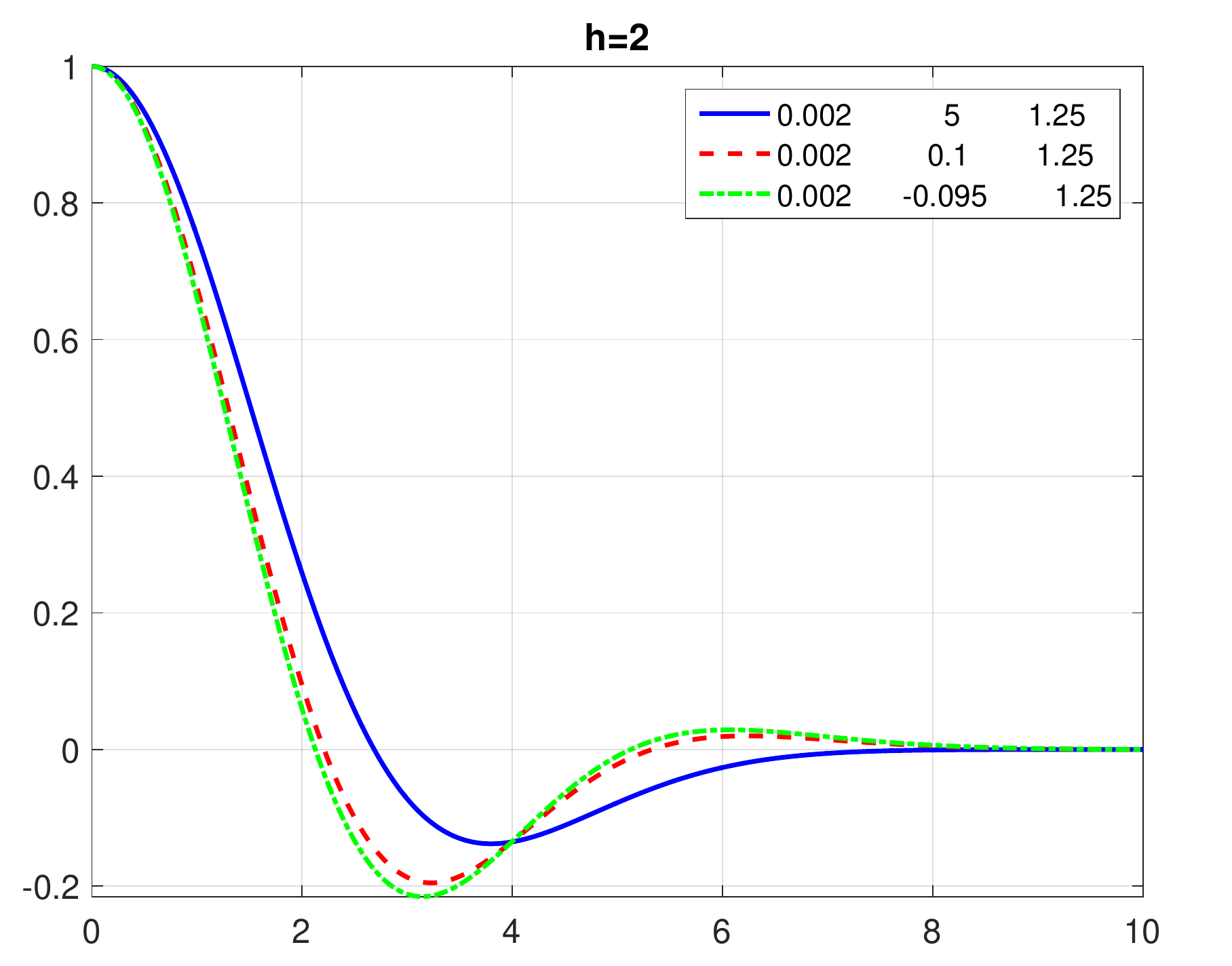}
\caption{Plot of the SPH-LAP2 precision function $Q^{\ast}(r;\bmthe,h)$ defined by Eq.~\eqref{eq:precision-squared-exponential} for $d=2$ and different parameter vectors $\bmthe$. Continuous line (blue): $\bmthe=(0.002, 5, 1.25)$.  Broken line (red): $\bmthe=(0.002, 0.1, 1.25)$. Dash-dot line (green): $\bmthe=(0.002, -0.095, 1.25)$.  The bandwidth of the squared exponential kernel is $h=1.5$.  All curves are normalized by dividing with $Q^{\ast}(0)$. }
\label{fig:fgc-sph}
\end{figure}

\section{The Continuum Case}
\label{sec:sph-lap2-continuum}

Theorem~\ref{theorem:sph-lap2-disc} presents the SPH-LAP2 joint pdf based on a discrete set of observations at the points in the set  $\Samp$. This can be extended to a continuum-space formulation by replacing the smoothed field of Eq.~\eqref{eq:x-sph} with the continuum convolution integral
\beq
\label{eq:x-sph-cont}
x^{\ast}(\bfs) = \int_{\Do} \D\bfs_{1}\, K(\bfs-\bfs_{1};h)\, x(\bfs_{1}).
\eeq
In contrast with Definition~\eqref{defi:sph-pd} (which is valid in the discrete case), in order to define the Gaussian measure in continuum space the points $\bfs$ should belong  to a bounded domain $\Do$  instead of $\Rd$---as shown in Theorem~\ref{theorem:gaussian}.

Then, the energy functional of Eq.~\eqref{eq:Ha-sph-lap2-disc} is replaced by

\beq
\label{eq:Ha-sph-lap2-cont}
\Ha[x({\Do}); \bmthe]=\frac{1}{2}\,
\int_{\Do}\D\bfs_{1}\,\int_{\Do}\D\bfs_{2}\,
x(\bfs_{1})\, Q^{\ast}(\bfs_{1}-\bfs_{2})\, x(\bfs_{2}),
\eeq
where the kernel $Q^{\ast}(\bfs_{1}-\bfs_{2};\bmthe,h)$ corresponds to the SPH-LAP2 precision function of Eq.~\eqref{eq:precision-sph}.
The integral formulation  extends the precision function~\eqref{eq:precision-sph} to  continuum,  bounded domains $\Do$.

\begin{remark}[Zero-bandwidth limit]
At the limit $h \to 0$ the FFT of the Gaussian kernel given by Eq.~\eqref{eq:fft-gauss} satisfies $\lim_{h \to 0 }\Kseft(k h) =1$. Hence, the  spectral function  Eq.~\eqref{eq:Qk-gauss} tends to the polynomial expression $\theta_{0} + \theta_{1} k^{2}   + \theta_{2} k^{4}$ which is the inverse spectral density (i.e., the image of the precision operator in Fourier space) of the FGC Boltzmann-Gibbs random fields (aka, Spartan random fields)~\cite{dth03,dthsel07}.
\end{remark}

\begin{figure}[tb]
\includegraphics[width=0.99\textwidth]{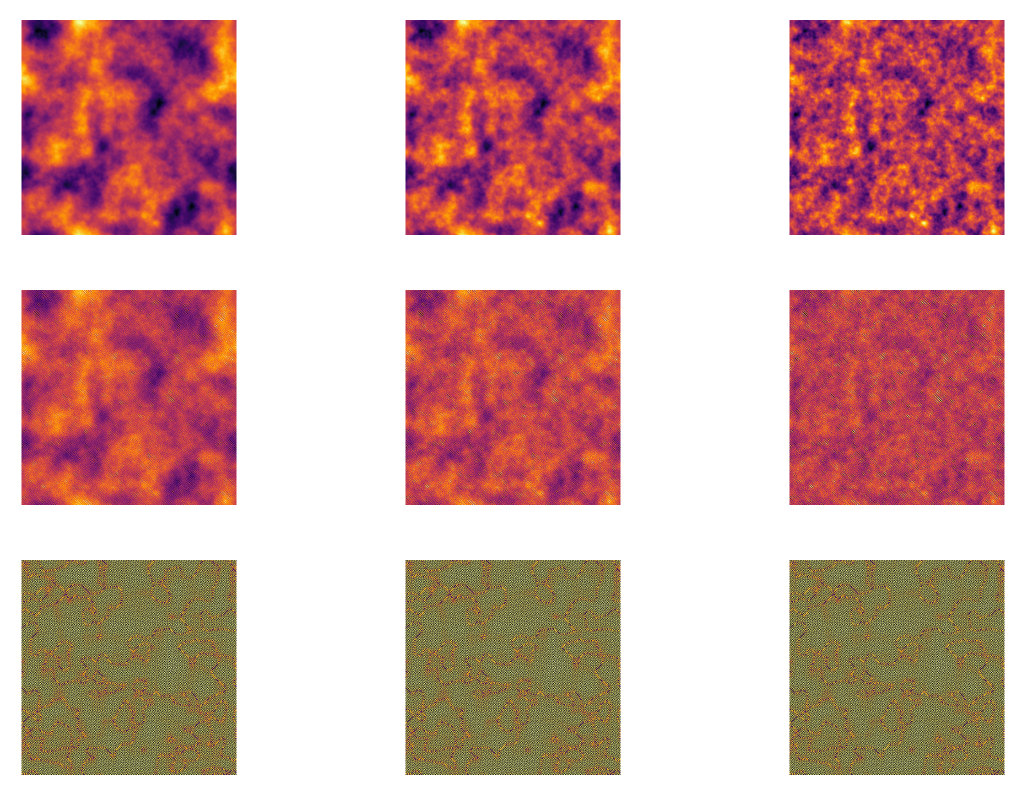}
\caption{Simulated states of SPH-LAP2 random fields with squared exponential kernel on $L \times L$ square grid ($L=256$). Top row: $h=0.05$. Middle row: $h=0.8$. Bottom row: $h=1.5$.  The field parameters are $\bmthe=\frac{1}{4\pi\xi^2} (1, 2\xi^2, \xi^4)^\top$: Left column:
$\xi=20$.     Middle column:  $\xi=10$.  Right column: $\xi=5$.}
\label{fig:sph-lap2-states}
\end{figure}

In the following, we use lattice simulations to visualize the Boltzmann-Gibbs SPH-LAP2  fields. The field states are generated using the spectral Fast Fourier Transform algorithm~\cite[pp.~701-708]{dth20}.
Figure~\ref{fig:sph-lap2-states} shows such  simulated states  on an $L \times L$ square grid ($L=256$) with unit lattice step for three different parameter vectors.  Different rows  (from top to bottom)  display random fields with $h=0.05, 0.8, 1.5$ respectively.  The parameter vector used is $\bmthe = \frac{1}{4\pi\xi^2} (1, 2\xi^2, \xi^4)^\top$.  This parametrization for $h=0$ implies that  Eq.~\eqref{eq:Qk-gauss} corresponds to $1/\tilde{C}(k)$, where $\tilde{C}(k)= 4\pi\xi^2/(1+ k^2\xi^2)^2$ is the spectral density of the Mat\'{e}rn covariance with smoothness index $\nu=1$, e.g.~\cite{Rue11,dth20}. The respective covariance function is given by $C(r)=(r/\xi) \,K_{1}(r/\xi)$, where $K_{1}(\cdot)$ is the modified Bessel function of the second kind of order $\nu=1$. The field configurations shown along each column of Fig.~\ref{fig:sph-lap2-states} differ with respect to $\xi$ which takes values $20, 10, 5$  (from left to right).  Lighter areas correspond to higher values, while darker areas to lower values of the field. All nine field states were generated using the same set of random numbers for comparison purposes.

The connection between the Mat\'{e}rn  $\nu=1$ and the SPH-LAP2 fields for  $h=0$ is illustrated by means of the variogram plots shown in Fig.~\ref{fig:sph-lap2-variograms}.  The variogram function is defined as $\gamma(\bfr)= \frac{1}{2} \mathrm{Var}[X(\bfs+\bfr;\om)- X(\bfs;\om)]$, where $\mathrm{Var}$ is the variance operator. Empirical variograms of the simulated states for $h=0.05$ are estimated along the rows and columns of the grid using the standard method-of-moments estimator~\cite{Cressie93}. The plots in Fig.~\ref{fig:sph-lap2-variograms} display averages of the row and column variograms and compare them with the $\nu=1$ Mat\'{e}rn model, $\gamma(r)=1 - (r/\xi) \,K_{1}(r/\xi)$.  The agreement is very good, especially for $\xi=5$, in spite of the finite ($h=0.05$) kernel bandwidth; for larger $\xi$ the observed small deviations between the Mat\'{e}rn model and the empirical curves are due to non-ergodic fluctuations of the field configurations caused by the higher  $\xi/L$ ratio.

\begin{figure}[tb]
\includegraphics[width=0.99\textwidth]{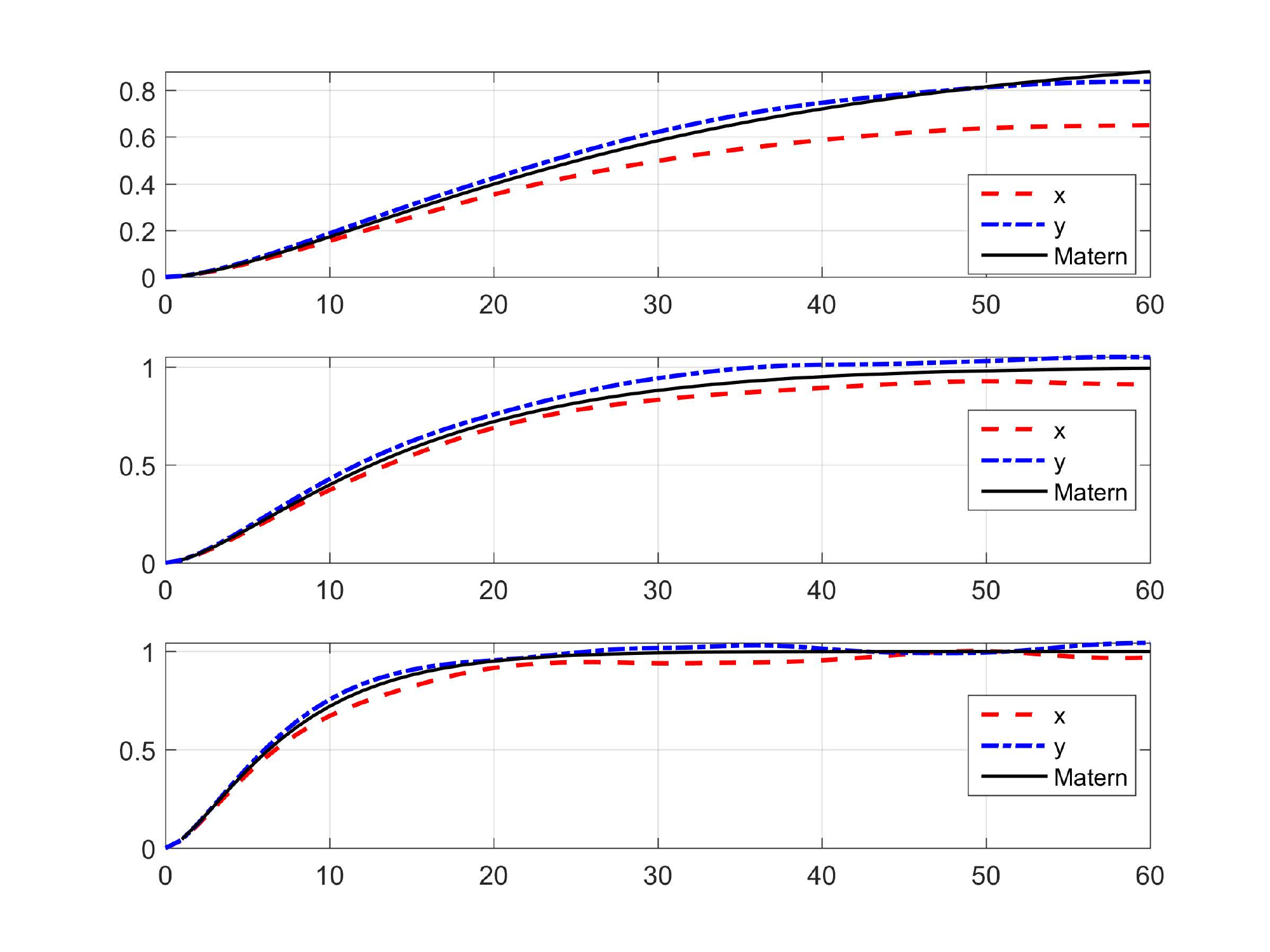}
\caption{Estimated variograms of the states shown in the first row of Fig.~\ref{fig:sph-lap2-states} (i.e., for $h=0.05$) and comparison with the Mat\'{e}rn variogram model $\gamma(r)=1 - (r/\xi) \,K_{1}(r/\xi)$ (continuous line). From top to bottom, $\xi=20, 10, 5$. Dash (dash-dot) lines show the average one-dimensional empirical variogram along the rows (columns) based on the method of moments.}
\label{fig:sph-lap2-variograms}
\end{figure}

\begin{remark}[Impact of kernel bandwidth]
The field configurations shown in Fig.~\ref{fig:sph-lap2-states} become increasingly rougher as $h$ increases. Intuitively this behavior can be understood as follows: Based on the Fourier transform of the precision function, Eq.~\eqref{eq:Qk-gauss}, higher $h$ implies that the spatial coupling  induced by the precision operator is strongly attenuated at high frequencies $k$ due to the factor $\exp(-k^2\,h^2/2)$. As a result, the spatial continuity imposed by the coupling  weakens at short distances (large $k$). This leads to sawtooth-like oscillations of the empirical variogram  which diminish with increasing lag (see Fig.~\ref{fig:sph-lap2-state-large-h}).
\end{remark}


\begin{remark}[Spatial patterns for  $h=1.5$]
The field states for $h=1.5$ (bottom row), exhibit a rather unusual behavior which is marked by alternating high and low values at neighboring sites and  ``islands'' of checkerboard patterns which correspond to variations in the magnitude of the fluctuations.  In addition, the impact of $\theta_{1}$ on the spatial patterns is not visible.  The main reason for this behavior is that values of $h>1$ emphasize the exponential decline of  the precision function and thus attenuate the coupling between neighboring sites.  A different way to interpret this pattern is by realizing that  $1/\tilde{Q}^{\ast}(k;\bmthe,h)$ has a sharp increase for large $k$ due to the factor $\exp(k^{2}h^{2}/2)$. This leads to large fluctuations of the field values at short distances. At low  $k$ such that $k\,h \ll 1$ (i.e., for large length scales) the impact of the exponential term is $\approx 1$, and the  polynomial term shapes the configuration of the ``islands''.
\end{remark}

In order to better visualize the spatial patterns of the field for $h=1.5$, Fig.~\ref{fig:sph-lap2-state-large-h} shows a spatial configuration generated on a smaller  grid ($L=64$).
As explained above, the SPH-LAP2 model tends to favor  opposite-sign values at neighboring sites.  It is thus similar to antiferromagnetic materials in which magnetic moments tend to align antiparallel to  neighboring moments~\cite{Neel48}. This is  evidenced in the sawtooth  oscillations of the variograms shown in Fig.~\ref{fig:sph-lap2-state-large-h}. However, in contrast to the classical Ising antiferromagnetic model which admits spin values equal to $\pm 1/2$~\cite{Ising25},  the SPH-LAP2 field is not limited to discrete values.

\begin{figure}[!ht]
\includegraphics[width=0.9\textwidth]{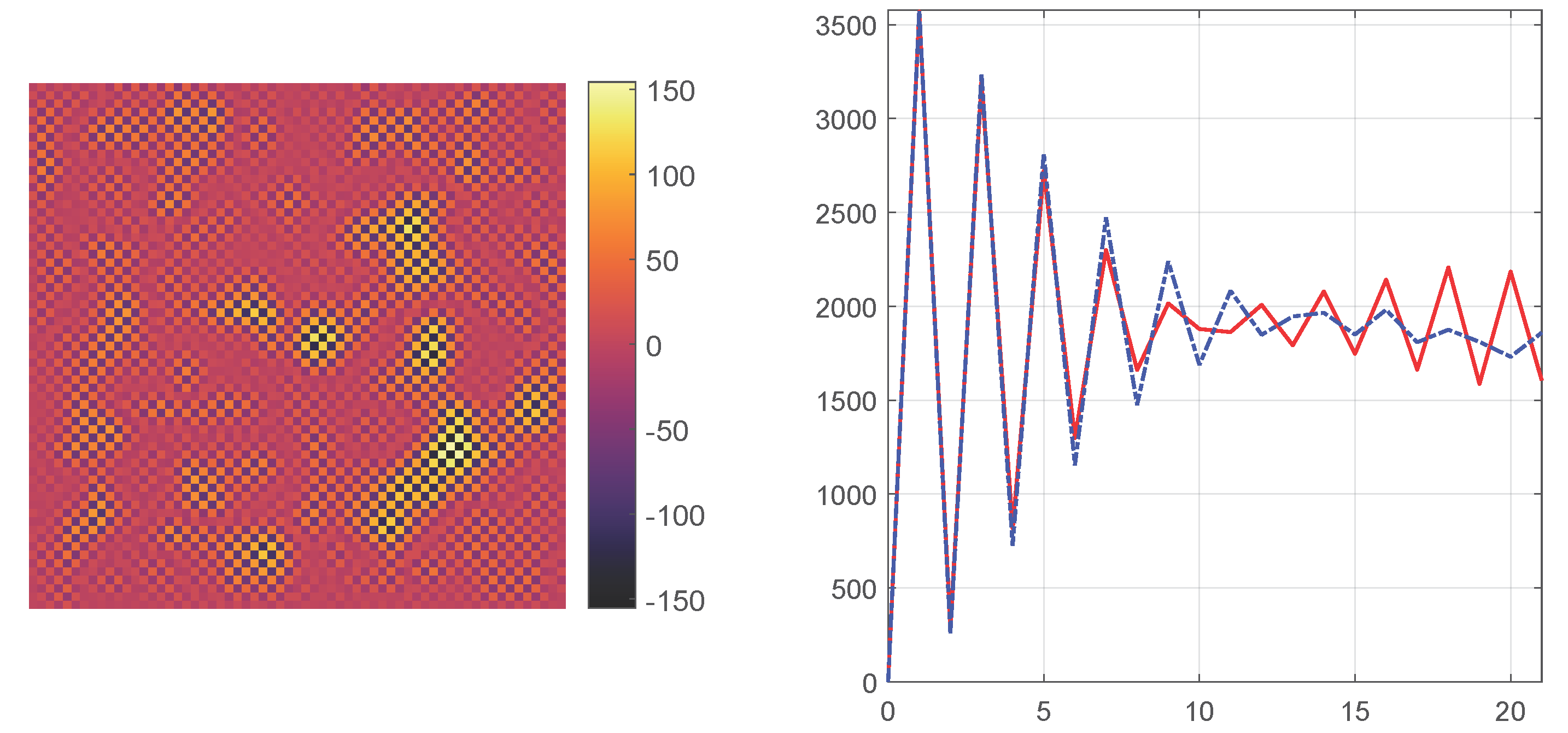}
\caption{Left: Simulated state  of SPH-LAP2 random field with squared exponential kernel on a square grid with $L=64$  nodes per side. Field parameters: $h=1.5$,  $\bmthe=\frac{1}{4\pi\xi^2} (1, 2\xi^2, \xi^4)$ with $\xi=20$.   Right: Average row and column variograms.   }
\label{fig:sph-lap2-state-large-h}
\end{figure}


\section{Discussion}
\label{sec:discussion}

The formulation presented herein was motivated by previous works by the author on Boltzmann-Gibbs random fields~\cite{dth03,dth20,ejs21} and the theory of smoothed particle hydrodynamics~\cite{Monaghan92,Monaghan05}. The present work combines a local precision operator constructed from the second-order polynomial of the Laplacian and smoothing kernel functions.  The link between covariance functions and polynomials of the Laplace operator has also been investigated in~\cite{dth03,Yaremchuk11}.  This is based on earlier ideas on the connections between covariance functions, rational spectral densities, Markov random fields and stochastic partial differential equations, e.g.~\cite{Rozanov77,Yaglom87,Rue11,dth20}.  Higher-order polynomials of the Laplace operator can be used to derive more general, SPH-LAP$n$  precision functions, where $n\in \Na$ and $n>2$.  Such precision functions can model more complex spectral properties, including non-monotonic dependence on $k$ which is relevant for wave phenomena. Another interesting direction of research is the investigation of  non-Gaussian smoothing kernel functions.  While the Gaussian kernel leads readily to explicit expressions,  it also imposes a drastic cutoff of the interactions at wavenumbers exceeding $\approx 1/h$.  Furthermore, the SPH-LAP2 function derived herein was based on uniform SPH weights $v_n$ in Eq.~\eqref{eq:Qnm}.  Richer, potentially non-stationary models can be achieved by allowing local variations of SPH weights.  Connections with machine learning models can be explored in this framework.

In the process convolution approach proposed by Higdon ~\cite{Higdon02},  observable random fields are defined via the convolution of a latent process---which is assumed to have a known covariance---with a kernel function.  The covariance of the  observed random field is then calculated by means of convolution integrals that involve the kernel function and the  latent process covariance.
In the SPH approach,  the latent field $\xk(\bfs)$ is defined by the weighted sum of the sample values given in Eq.~\eqref{eq:x-sph}, where the weights are provided by the kernel function.  However, the  covariance function of $\xk(\bfs)$ is not known. Instead, the values of $\xk(\bfs)$ are coupled to each other via differential operators, which enforce the spatial correlations.  Hence, the SPH representation focuses by construction on the precision operator. Another difference between the Higdon and the SPH representations is that in the former the observable process is defined by  the convolution of a kernel function with a latent field which is not directly observable, while the SPH convolution in Eq.~\eqref{eq:x-sph}, and its continuum analogue Eq.~\eqref{eq:x-sph-cont},  define the latent field $\xk(\bfs)$ in terms of the observed field values. The Higdon latent field can comprise a set of independent, identically distributed random variables; then, the covariance of the observable process is determined purely from the kernel function. This is not possible in the SPH representation, where the latent field is determined by  the kernel function and the sampled process which is spatially correlated.

The class of SPH-LAP2 random fields   can find applications in spatial statistics, statistical physics, and  machine learning.  In spatial statistics and machine learning the benefit of the model derives from the explicit form of the precision function. This permits the computationally fast estimation of the single-point predictive distribution even for  big datasets. The efficiency of the method is due to the fact that the inversion of the large $N \times N$ Gram matrix, necessary for covariance-based methods, is avoided.

In particular,  using standard expressions for GMRFs~\cite{Rue05,dth20},  the predictive pdf at any location $\bfs \in \Do$ is given by the following normal density

\begin{subequations}
\begin{align}
\ff\left(x_s \mid \bfx \right) =  & {\mathcal N}(m_{x_s \mid \bfx}, \sigma^{2}_{x_{s} \mid \bfx})
\\
m_{x_s \mid \bfx} = & - \sum_{n=1}^{N} \frac{Q^{\ast}(\bfs- \bfs_{n};\bmthe,h)}{Q^{\ast}(0;\bmthe,h)}\, x_{n},
\\
\sigma^{2}_{x_{s} \mid \bfx} = & \frac{1}{Q^{\ast}(0;\bmthe,h)},
\end{align}
\end{subequations}
where $m_{x_s \mid \bfx}$ and $\sigma^{2}_{x_{s} \mid \bfx}$ are, respectively, the conditional mean and variance given a vector of sample values $\bfx$ at the set  of sampling positions $\Samp$.

For applications  in statistical physics, the model enhances the standard Gaussian field theory~\cite{Kardar07,Goldenfeld93} by including the curvature term in the energy functional and introducing a gradual decline of the interactions with distance.  The correlation function and the free energy of this model are topics that deserve further investigation.
Notice that the curvature term is irrelevant for  long wavelengths (small $k$), which is the interesting limit in the study of phase transitions~\cite{Kardar07}.  However, the curvature term is pertinent in  Ginzburg-Landau theories of physical systems such as ternary amphiphilic mixtures~\cite{Gompper93a,Gompper93}. In such systems one is interested in the geometry of random field configurations; therefore the spectral properties at  intermediate wavelengths are also important.  Finally, it would be interesting to investigate  connections between  SPH-LAP2  and magnetic models.

The SPH-LAP2 model is expressed in terms of the precision function given by Eq.~\eqref{eq:precision-sph}. This leads to sparse precision matrices for compactly supported kernel functions.  As mentioned in the above paragraph, the ensuing formulation has a direct connection with GMRFs.  These properties (sparse precision matrix and GMRF connection) are shared by the popular stochastic partial differential equation (SPDE) approach~\cite{Rue11}, recently reviewed in~\cite{Lindgren22}.
Note that the BG-LAP2 model defined in Section~\ref{sec:gauss-bg} can be derived from an SPDE with a suitable differential operator which involves the first and the second power of the half-Laplacian~\cite{dth20}.  However, the realizations of BG-LAP2 do not admit pointwise derivatives (for $d>1$). In $d=1$ the BG-LAP2 model represents  a linear harmonic damped oscillator driven by white noise~\cite[Chap.~9]{dth20}. On the other hand, random fields with Mat\'{e}rn covariance,  obtained from the classic SPDE approach~\cite{Rue11}, admit (in mean square sense) partial derivatives up to order $m = \lfloor \nu \rfloor$, where $\nu$ is the smoothness parameter.
In the SPH-LAP2 formulation, the high-frequency interactions are cut off by the smoothing kernel [cf.~\eqref{eq:precision-function-ift}].  This leads to sharp transitions of the field for finite $h$ (cf. Fig.~\ref{fig:sph-lap2-states}). A comparison between the SPDE and SPH approaches can be carried further with respect to non-stationarity, non-Gaussian  models, parameter inference, and computational efficiency. However, this is beyond the scope of the current paper.  To summarize, the BG-LAP2 model can be derived from an SPDE, which is in general different than the Mat\'{e}rn-related SPDE; the  two models become equivalent only in $d=2$ if  $\nu=1$ (Mat\'{e}rn) and $\bmthe \propto (\xi^{-2}, 2, \xi^{2})^\top$ (BG-LAP2).  The SPH-LAP2 model is not at this point linked with an SPDE, but it provides (at least in the case of a squared-exponential smoothing kernel) an explicit precision function given by Eq.~\eqref{eq:precision-squared-exponential}.

\section{Conclusions}
\label{sec:conclusions}


This paper presents  a novel class of Gaussian Boltzmann-Gibbs random fields defined by means of a mesh-free spatial coupling (precision) function determined by Eq.~\eqref{eq:precision-function-ift}.  The latter is based on a smoothing kernel and a second-degree polynomial of the  Laplace operator.  The  smoothing kernel  is a key factor for  the mesh-free representation which can be applied to any spatial configuration of  sampling points. The SPH representation allows constructing Boltzmann-Gibbs energy functionals, such as Eq.~\eqref{eq:Ha-fgc-kernel}, that involve differential operators, even for discretely sampled data. Thus, the SPH-LAP2 model bridges the discrete (``particle'') and continuum (``wave'') representations. It extends the concept of Gaussian Markov random fields to a mesh-free formulation that can be applied in both continuous domains and scattered data. An explicit expression, given by Eq.~\eqref{eq:precision-squared-exponential}, is derived for the precision function using the Gaussian (squared exponential) smoothing kernel.



\end{document}
